%% file: hal.tex
\documentclass[a4paper,UKenglish]{article}

\title{Absolute convergence and Taylor expansion in web based 
models of Linear Logic} 

\date{}

\author{Christine Tasson \\ 
Fédération ENAC ISAE-SUPAERO ONERA, Université de Toulouse, France.
\and 
Aymeric Walch \\ 
Fédération ENAC ISAE-SUPAERO ONERA, Université de Toulouse, France.}

\usepackage[utf8]{inputenc}
\usepackage[T1]{fontenc}

\usepackage{notations}

\usepackage[total={5.5in, 8.5in}]{geometry}
\usepackage{microtype}

\usepackage{amssymb} 
\usepackage{amsthm}
\usepackage{mathtools} 

\usepackage{thm-restate} 

\usepackage{hyperref}

\usepackage{cleveref}

\theoremstyle{plain}
\newtheorem{theorem}{Theorem}
\newtheorem{lemma}[theorem]{Lemma}
\newtheorem{corollary}[theorem]{Corollary}
\newtheorem{proposition}[theorem]{Proposition}
\newtheorem*{theorem*}{Theorem}

\theoremstyle{definition}
\newtheorem{definition}[theorem]{Definition}
\newtheorem*{definition*}{Definition}

\newtheorem{remark}[theorem]{Remark}
\newtheorem{example}[theorem]{Example}

\usepackage[pdftex,dvipsnames]{color} 
\usepackage{svg}
\usepackage{pdfpages}

\definecolor{blue(pigment)}{rgb}{0.2, 0.2, 0.6}
\definecolor{darkgreen}{rgb}{0.0, 0.5, 0.0}
\definecolor{darkred}{HTML}{9F000F}
\definecolor{darkdarkred}{HTML}{4A0005}

\definecolor{deepblue}{rgb}{0,0,0.5}
\definecolor{deepred}{rgb}{0.6,0,0}
\definecolor{deepgreen}{rgb}{0,0.5,0}

\hypersetup
{
     colorlinks=true,
     linkcolor=darkdarkred, 
     filecolor=magenta,            
     urlcolor=blue(pigment),
     citecolor=darkgreen
}

\makeatletter
\newcommand{\labeltext}[2]{%
  \@bsphack
  \csname phantomsection\endcsname 
  \def\@currentlabel{#1}{\label{#2}}%
  \@esphack
}
\makeatother

\usepackage{tikz-cd}
\tikzcdset{row sep/abysmal/.initial=0.1em} 
\usepackage{quiver}

\begin{document}

\maketitle

\paragraph*{Abstract.}
The differential $\lambda$-calculus studies
  how the quantitative aspects of programs correspond to 
  differentiation and to Taylor expansion inside 
  models of linear logic.
  Recent work has generalized the axioms of 
  Taylor expansion so they apply to many models that only feature  
  partial sums. However, that work does not cover the classic 
  web based models of Köthe spaces and finiteness spaces.
  
  First, we provide a generic 
  construction of web based models 
  with partial sums. It
  captures  models,
  ranging from coherence spaces to probabilistic coherence spaces, 
  finiteness spaces and Köthe spaces.
  Second, we generalize the theory of Taylor expansion 
  to models in which coefficients can be non-positive. 
  We then use our generic web model construction 
  to provide a unified proof that all the aforementioned web based 
  models feature such Taylor expansion.


\paragraph*{Acknowledgements.} 
We would like to thank Thomas Ehrhard for his helpful 
advices.
Our notion of absolute PCR was created following one of his suggestions.

\input{content.tex}

\end{document}

%% file: content.tex
\section*{Introduction}

%
Quantitative semantics was introduced with linear logic
 after Girard noticed that the interpretation of programs in mathematical models,
were \emph{analytic maps}~\cite{Girard87}. They can be written as an infinite
sum of multilinear maps through a formula called the 
%
Taylor expansion:

 $\displaystyle f(x) = \sum_{n \in \N} \frac{1}{\factorial n} \hod{f}{n}(0) \cdot (x, \ldots, x)\qquad $
  where $\hod f n(0)$ is the $n$-th derivative of $f$ at $0$.\\
The $n$-th derivative at $0$ captures the 
part of the program that uses its input \emph{exactly} $n$-times 
during execution.
%
In this way, quantitative semantics  helps us understand 
 how computation interacts with 
non-determinism \cite{LairdManzonettoMcCuskerPagani13}, 
probabilities \cite{DanosEhrhard11,EhrhardGeoffroy25}
or quantum primitives \cite{PaganiSelingerValiron14}.

The differential $\lambda$-calculus \cite{EhrhardRegnier02} and the syntactic 
Taylor expansion \cite{EhrhardRegnier08} extend quantitative 
semantics. They give a purely syntactic way to track resources
by building  derivative and Taylor expansion inside the language itself.
These theories  also uncover a strong link between the world of logic and computation,
and the world of topology and analysis 
\cite{Ehrhard02,Ehrhard05,KerjeanTasson18,KerjeanPedrot24}.
%

The differential $\lambda$-calculus
\cite{BucciarelliEhrhardManzonetto10} and the syntactical 
Taylor expansion, and their model \cite{Manzonetto12}  have a fatal 
drawback: they allow  any \emph{countable} sums. The only way to interpret
such sums computationally is as \emph{non-determinism}.
From a syntactic viewpoint, this is puzzling.
Indeed, Taylor expansion can be applied to deterministic~\cite{EhrhardRegnier08}
or probabilistic~\cite{DalLagoLeventis19}
programs, and enjoys \emph{uniformity} properties 
that are needed to prove normalization~\cite{EhrhardRegnier08,Vaux17}.
From a semantic viewpoint, allowing countable sums rules out most quantitative 
semantics. 
Fortunately, recent works show that the differential 
calculus \cite{Ehrhard23-cohdiff,EhrhardWalch25} and the 
Taylor expansion \cite{EhrhardWalch25} can be axiomatized using only 
\emph{partial sums}. Partial sums 
match computational behaviours such as 
determinism in coherence spaces~\cite{Girard87}, or randomness
in probabilistic coherence spaces~\cite{DanosEhrhard11}.

One limitation remains. The infinite partial sums used in the \emph{coherent} 
theory of Taylor expansion \cite{EhrhardWalch25} 
rely on $\Sigma$-monoids \cite{Haghverdi00}.
$\Sigma$-monoids only allow positive sums: if $x + y = 0$ then $x = y = 0$.
For programming language semantics, this is not a problem as programs never cancel each other out.
However, this restriction becomes a serious issue when developing  
a theory of Taylor expansion that also works  for quantitative 
semantics built on analysis and topology. In particular, 
Köthe spaces~\cite{Ehrhard02} and finiteness spaces \cite{Ehrhard05}
are important quantitative models that use analytic maps. Because 
they involve negative coefficients, they are  not covered by the coherent 
Taylor expansion described in \cite{EhrhardWalch25}.

\subparagraph*{Contributions.} 
Our main contribution (\Cref{part:taylor}) is to generalize the 
coherent Taylor expansion~\cite{EhrhardWalch25} to a more general notion of partial sums called 
PCMs~\cite{Hines13} that includes negative coefficients.
%
This axiomatization of partial 
sum  is more subtle to work with, 
in the same way that summability over real numbers is more subtle than 
summability over non-negative real numbers. 
In order to prove in a \emph{unified} way 
that both Köthe spaces and finiteness spaces feature a coherent Taylor expansion,
we introduce a generic construction of web models (\Cref{part:web-models}), which 
is another contribution.
%
%
We first consider models in which the coefficient of 
the vectors and matrices range over an arbitrary $\Sigma$-monoid
in which objects are sets of vectors and morphisms are matrices.
We then generalize our construction to arbitrary PCMs whose summability is 
given by a notion of absolute convergence. This 
describes web models where coefficients can be negative, as in  
Köthe spaces \cite{Ehrhard02} and
finiteness spaces \cite{Ehrhard05}.
%
%
\subparagraph*{Related work.}
The use of partial sums in semantics is rather common, 
and can be traced back to \emph{algebraic programming} 
semantics \cite{ManesArbib86}. Various notions of positive 
sums have been used in the literature, such as 
in geometry of interaction~\cite{Haghverdi00} and quantum programming 
semantics~\cite{Selinger04,ChoJacobsWesterbaan15}.
Our results on web models provide a unified point of view on various proof 
techniques that appear throughout the literature on web based 
semantics~\cite{Girard87,Ehrhard02,Ehrhard05,DanosEhrhard11,LairdManzonettoMcCuskerPagani13}.
%
%
Under some light additional assumptions (see \cref{rem:semimodules}), our web 
models based on $\Sigma$-monoids 
can be described as categories of modules over $\Sigma$-semirings~\cite{TsukadaKazuyuki22}
with an orthogonal basis reminiscent of tight orthogonality in double glueing~\cite{HylandSchalk03}.
However, \cite{TsukadaKazuyuki22} does not describe 
Taylor expansion in that setting, and even though 
they prove that their model have a Lafont exponential, they do not provide any 
explicit definition. They also do not provide a definition for 
more general PCMs based on absolute convergence. 
Our theory of coherent Taylor expansion is an adaptation of 
the coherent Taylor expansion of~\cite{EhrhardWalch25}
to a setting with possibly negative coefficients.
Our work lies in the field of categorical differentiation, whose aim 
is to study the categorical properties of 
differentiation~\cite{BluteCockettSeely06,BluteCockettSeely09,CockettCruttwell14} 
and of Taylor expansion~\cite{KerjeanLemay23,Lemay24,Walch25}.
\subparagraph*{Notations.}
Let $\catLL$ be a
model of \textbf{linear logic}.
We write the composition of $f \in \catLL(X, Y)$ with 
$g \in \catLL(Y, Z)$ as $g \compl f$.
We write the symmetric monoidal structure as $(\sm, \smunitL, \smunitR, \smassoc, 
\smsym)$ where $\smunitL_X : \smone \sm X \arrow X$ and 
$\smunitR_X : X \sm \smone \arrow X$ are the unitors, 
$\smassoc_{X, Y, Z} : (X \sm Y) \sm Z \arrow X \sm (Y \sm Z)$
is the associator, and $\smsym_{X, Y} : X \sm Y \arrow Y \sm X$ is 
the symmetry.
We write $X \linarrow Y$ the internal hom, the currying of 
$f : X \sm Y \arrow Z$ is $\cur(f) : X \arrow (Y \linarrow Z)$, 
and the evaluation is $\ev : (X \linarrow Y) \sm X \arrow Y$.
There is a bifunctor 
$\_ \linarrow \_ : \catLL^{op} \times \catLL \arrow \catLL$
that maps $f : X' \arrow X$ and $g : Y \arrow Y'$ to 
$f \linarrow g = \cur(g \compl \ev \compl \left( (X \linarrow Y) \sm f \right)) : 
(X \linarrow Y) \arrow (X' \linarrow Y')$.
Let $(\bang, \der, \dig)$ be the ressource comonad,
$\klexp$ the Kleisli category of $\bangtext$, and 
$\Der : \catLL \arrow \klexp$ the canonical functor 
that maps $f \in \catLL(X, Y)$ to $f \compl \der \in \catLL(\bang X, Y)$.

We write $I$-indexed product as $\withFam X_i$, the projections as 
$\proj_i : \withFam X_i \arrow X_i$. 
We write $I$-indexed coproducts as 
$\plusFam X_i$, the injections as
$\inj_i : X_i \arrow \plusFam X_i$. We write the paring and co-pairing 
of a family $f_i : X_i \arrow Y$ as
$\prodPairing{f_i} : X \arrow \withFam Y_i$ and
$\coprodPairing{f_i} : \plusFam X_i \arrow Y$. 

If $A$ is a countable set and $M$ is a set, an $A$-\textbf{indexed family} of 
elements of $M$ is a function $\vect x : A \arrow M$, written 
$\vect x = \family{x_a}$. We define 
the \emph{support} of a family $\vect x = \family{x_a}$ as 
$\supp{\vect x} = \{a \in A | x_a \neq 0 \}$.
A finite \textbf{multiset} over a set $A$ is a function $m : A\arrow \N$ 
with finite support.
Let $\mfin{(A)}$ be the set of all finite multisets over $A$.
We also use the notation 
$[a_1, \ldots, a_n]$ for the 
function that maps $a$ to the size of 
$\set{i \suchthat a = a_i}$.
We write $m_1 + m_2$ their pointwise sum, so that 
$[a_1, \ldots, a_n] + [a_1', \ldots, a_k'] = [a_1, \ldots, a_n, a_1', \ldots, a_k']$.

A category $\catLL$ has \textbf{zero morphisms} if for all $X, Y$, there 
exists a morphism $0^{X, Y} \in \catLL(X, Y)$ 
such that for all $f \in \catLL(X', X)$ and $g \in \catLL(Y, Y')$, 
$0^{X, Y} \compl f = 0^{X', Y}$ and  $g \compl 0^{X, Y} = 0^{X, Y'}$.
We  define, for all object $X$, 
the \textbf{Kronecker} symbole $\kronecker i j \in \catLL(X, X)$
by $\kronecker i i = \id$ and $\kronecker i j = 0$ if 
$i \neq j$.

\section{Partial Commutative Monoids (PCMs)} 

\label{sec:pcm}

We need a suitable axiomatization of partial sums
that captures sums that are typical to analysis and topology.
We use the notion of \emph{partial commutative monoid}~\cite{Hines13}.
We consider a set $M$  with a partial function 
$\Sigma$ called the \emph{sum} from indexed family on $M$ to $M$. An indexed family 
$\family{x_a}$ is \emph{summable} if it is in the domain of $\Sigma$,
and we write its image as $\sum_{a \in A} x_a$.

\begin{definition} \label{notation:sum-def}
  For any  expression
  $e$ and $e'$ involving partial sums, we write 
  $e \sumsub e'$ if whenever $e$ is defined, then $e'$ is also defined 
  and $e = e'$. We write  $e \sumiff e'$ if $e \sumsub e'$ and $e \sumsubinv e'$.
\end{definition}

\begin{definition}[\cite{Hines13}] 
    \label{def:pcm} 
    \labeltext{(Unary)}{ax:unary}
  \labeltext{(WPA)}{ax:wpa}
  The tuple $(M, \Sigma)$ is a \emph{partial commutative monoid (PCM)} if $M$ 
  is not empty and if the
  sum $\Sigma$ satisfies the following axioms. \begin{itemize}
    \item \ref{ax:unary} Unary sum axiom: every singleton family $(x)$ is summable 
    with sum $x$.
    \item \ref{ax:wpa} Weak partition associativity axiom: 
    Let $\family{x_a}$ be an indexed family, $I$ countable, and $\{A_i\}_{i \in I}$ a partition 
    of $A$ (the $A_i$ can be empty). 
    Then $\sum_{a \in A} x_a \sumsub 
      \sum_{i \in I} \left(\sum_{a \in A_i} x_a \right)$
  \end{itemize}
  A PCM is \emph{strong} if it satisfies 
  the partition associativity axiom \labeltext{(PA)}{ax:pa} \ref{ax:pa},
  where the implication~$\sumsub$ of \ref{ax:wpa} is replaced by an equivalence 
  $\sumiff$. 
  Strong PCMs are often  called $\Sigma$-monoids~\cite{Haghverdi00}.
\end{definition}

Observe that the neutral element $0$ is not given as part of the
data, but can be defined from the axioms. Indeed, the sum over the 
empty family is always defined by \cref{prop:subfamily-summable}.

\begin{lemma} \label{prop:subfamily-summable}
  If $A \subseteq B$ and $\family<B>[b]{x_b}$ is summable, then 
  $\family{x_a}$ is summable.
\end{lemma}


\begin{definition} \label{def:preorder-sum}
Every PCM has a canonical preorder (it is not antisymetric 
in general) given by 
$x \leq y \text{ if there exists } z \text{ such that } x + z = y$.
\end{definition}

Algebraic programming semantics
is often based on strong PCMs 
and \emph{additive domains}~\cite{ManesArbib86}.
An additive domain is (roughly) a \emph{strong} PCM where the preorder 
of \cref{def:preorder-sum} is antisymetric, and where
$\sum_{a \in A} x_a = \sup \set{\sum_{a \in F} x_a 
\suchthat F \subseteq A \text{ finite}}$.
However, strong PCMs 
$(M, \Sigma)$ are \emph{positive}~\cite{Hines13}: for all $x, y \in M$, 
$x + y = 0 \imply x = y = 0$, as illustrated 
in \cref{ex:real,ex:finitary}.
Positivity is not an issue 
when one is concerned with programming languages semantics.
However, it does not capture partial summability in model based on traditional analysis.

\begin{example}\label{ex:real}
  A family of \emph{non-negative} real numbers 
  $\family{x_a \in \R_{\geq 0}}$
  is summable when
  $D = \{ \sum_{a \in F} x_a | F \text{ is a finite subset of } A \}$ 
  has an upper bound, written $\sum_{a \in A} x_a = \sup{D}$.
  This sum is an additive domain.
  A countable family of real numbers $\family{x_a \in \R}$ is 
  \emph{absolutely convergent}
  if the family of their absolute 
  values $\family{\abs{x_a}}$ is summable. 
  Then we can define the sum of the $\family{x_a}$ as
  $\sum_{a \in A} x_a \defEq \sum_{a \in A} \absplus x_a - \sum_{a \in A} \absminus x_a$,
  where $\absplus x_a = \max(0, x_a)$ 
  and $\absminus x_a = \max(0, -x_a)$. 
  The notion of absolute convergence yields a PCM, the axiom \ref{ax:wpa} 
  corresponds to the Fubini-Tonelli 
  theorem. However, this PCM is not strong: 
  $(-1, 1)$ is summable with sum $0$, $(0, 0, \ldots)$ is summable, but 
  $(-1, 1, -1, 1, \ldots)$ is not.
\end{example}

\begin{example} \label{ex:finitary}
For all monoid $M$, we can define the \emph{finitary PCM} generated by $M$ 
in which a family $\vect x = \family{x_a}$ is summable if 
$\supp{\vect x} = \set{a \in A \suchthat x_a \neq 0}$ is finite, with 
sum defined as 
$\sum_{a \in A} x_a = \sum_{a \in \supp{\vect{x}}} x_a$.
This PCM is strong if and only if $M$ is positive.
\end{example}
%


%
\subparagraph*{Properties of PCMs.}
First, the sum over the empty family, 
$0$, is a neutral element.

\begin{proposition} \label{prop:zero-neutral}
  Let $\vect x = \family{x_a}$ be an indexed family, and $A'$ be a set such that 
  $\supp{\vect x} \subseteq A' \subseteq A$.
  Then $\sum_{a \in A'} x_a \sumiff \sum_{a \in A} x_a$. 
\end{proposition} 

\begin{proof}
  Assume that $\family<A'>{x_a}$ is summable. Then for all $a \in A$, define 
  $A_a = \{a\}$ if $a \in A'$, and $A_a = \emptyset$ otherwise. 
  Then, $\{A_a\}_{a \in A}$ is a partition of $A'$. 
  Thus, by \ref{ax:wpa}, 
  $\sum_{a \in A'} x_a \sumsub \sum_{a \in A} \sum_{a' \in A_a} x_{a'}$.
  But for all $a \in A$, 
  $\sum_{a' \in A_a} x_{a'} = x_a$, by \ref{ax:unary} if $x_a \neq 0$, 
  and by definition of $0$ if $x_a = 0$.
  Thus, $ \sum_{a \in A'} x_a \sumsub \sum_{a \in A} x_a$. 
  Conversely, if $\family{x_a}$ is summable then $\family<A'>{x_{a}}$ is summable 
  by \cref{prop:subfamily-summable}
  because it is a subfamily of $\family{x_a}$, and it is proved above that the sums are equal.  
\end{proof}

Sums in PCMs are agnostic to the index set.
%
%
Given two index sets $A$ and $B$, an injection \(\phi:A \injection B\) 
and a family
$\vect x= \family{x_a}$ we define a $B$-indexed family
$\famact\phi \vect x= \family<B>[b]{y_b}$ by
$y_{\phi(a)}= x_a$, and $y_b = 0$ if $b\notin\phi(A)$.


\begin{proposition} \label{prop:reindexing}
For any injection $\phi : A \injection B$ and any $A$-indexed family $\vect x$,
$\famact{\phi} \vect x$ is summable if and only if $\vect x$ is summable, 
and the two sums are equal.
\end{proposition}

\begin{proof}
  Let  $\vect y = \famact{\phi} \vect x = \family<B>[b]{y_b}$.
  By definition, $\supp{\vect y} \subseteq \im(\phi) \subseteq B$, 
  so by \cref{prop:zero-neutral}, $\vect y$ 
  is summable if and only if $\family<\im(\phi)>[b]{y_b}$
  is summable, and the two sums are equal. 
  Thus, it suffices to prove our result for all bijections to conclude 
  that it holds for all injections.
  We now assume that $\phi$ is a bijection.
  Assume that $\vect x$ is summable.
  For all $b \in B$, let $A_b = \{\phi^{-1}(b)\}$. Then the $\{A_b\}_{b \in B}$
  are a partition of $A$, so by \ref{ax:wpa} we have
  $\sum_{a \in A} x_a \sumsub \sum_{b \in B} \sum_{a \in A_b} x_a$.
  By \ref{ax:unary}, $\sum_{a \in A_b} x_a = \famact{\phi} \vect x (b)$. 
  Thus, $\famact{\phi} \vect x$ is summable 
  with the same sum as $\vect x$. 
  Conversely, if $\famact{\phi} \vect x$ is summable then 
  $\famact{(\phi^{-1})} \famact{\phi} \vect x = \vect x$ is summable. 
\end{proof}
\subparagraph*{Partial commutative rigs.}
We introduce a PCM counterpart to the notion of rigs 
(also called commutative semirings)~\cite{Glazek02}. 
They will serve as the coefficients of our web models~(\cref{part:web-models}).

\begin{definition}
A Partial Commutative Rig (PCR) is the data of $(\rig, \Sigma, 1, \cdot)$
such that $(\rig, \Sigma)$ is a PCM, 
$(\rig, 1, \cdot)$ is a commutative monoid, and such that 
for all indexed families 
$\family{x_a}$ and $\family<B>[b]{y_b}$ of $\rig$,
$\left(\sum_{a \in A} x_a \right) \cdot \left(\sum_{b \in B} y_b \right)
\sumsub \sum_{a \in A, b \in B} (x_a \cdot y_b)$.
A PCR is strong if the underlying PCM is strong. Strong PCR are 
called $\Sigma$-semirings in~\cite{TsukadaKazuyuki22}.
\end{definition}


\begin{example} \label{ex:pcr}
\begin{itemize}
\item Both $\R$ and $\Rpos$ are PCR, with the PCM structure of
\cref{ex:real} and with the usual multiplication.
\item Every rig $(\rig, 0, +, 1, \cdot)$
induces a PCR called the \emph{finitary PCR} generated by $\rig$,
in which the PCM structure is the finitary PCM induced by $(\rig, +)$
(see \cref{ex:finitary}).
\item A complete rig is a PCR in which sums are always defined. A \emph{continous 
rig} \cite{LairdManzonettoMcCuskerPagani13} is a complete 
rig in which the PCM structure is an additive domain.
A standard example of complete rig is the completion of positive real 
numbers $\Rbar = \Rpos \cup \set{\omega}$.
\end{itemize}
\end{example}

\smallskip
\noindent 
\textbf{PCM categories.}
We now give a definition of categories ``enriched'' over PCM.
This enrichement is at the core of our theory of Taylor expansion, in the 
same way that differential categories \cite{BluteCockettSeely06}
are enriched over commutative monoids.
Let $\catLL$ be a locally small category such that 
for all objects $X, Y$, $\catLL(X, Y)$ is a PCM.
%

\begin{definition}[\cite{Hines13}] \label{def:pcm-category}
  The category $\catLL$ is a 
  PCM-category if for all 
$\family<A>[a]{f_a \in \catLL(X, Y)}$ and $\family<B>[b]{g_b \in \catLL(Y, Z)}$,
$\left( \sum_{b \in B} g_b \right) \compl \left( \sum_{a \in A} f_a \right) 
\sumsub  \sum_{(a, b) \in A \times B} (g_b \compl f_a)$.
\end{definition}

This property is called the strong distributivity of the composition.
By \ref{ax:unary}, strong distributivity implies both left and right distributivity:
$\left( \sum_{b \in B} g_b \right) \compl f \sumsub \sum_{b \in B} (g_b \compl f)$
and $g \compl \left( \sum_{a \in A} f_a \right) \sumsub \sum_{a \in A} (g \compl f_a)$.
However, left and right distributivity do not imply strong 
distributivity, unless the PCM is strong \cite{Hines13}. We need to be very careful
about this subtlety.
It follows from left and right distributivity
that a PCM-category always have $0$-morphisms, defined as the sum over 
the empty family.

\part{A generic construction of web models}

\label{part:web-models}

We provide in this part a unified construction of web models of linear 
logic, based on a generic PCR $\rig$. 
This definition encompasses many web based models 
of the literature depending on the choice of $\rig$ (see \cref{ex:web-models-positive}), and clearly identify
what common patterns these models share. 
We give two constructions. The first one follows the standard construction of most 
web based models, but only works when the PCR is \emph{strong}.
The second one works on more general PCR, whose summability consists of 
some sort of \emph{absolute convergence}.

Let us introduce new notations for this part.
We define the Kronecker symbol: $\kronecker a b = 0$ if $a \neq b$, and 
$\kronecker a a = 1$. We define for all $a \in \web{X}$ a vector 
$e_a \in \vrig X$ by $(e_a)_{a'} = \kronecker a {a'}$. 
For all set $\web X$ (called the \emph{web}) and all PCR $\rig$, we define the following 
\emph{partial} maps on vectors: 
\begin{itemize}
\item scalar product $\scalarRig{.}{.} : \vrig{X} \times 
\vrig{X} \arrow \rig$ by  $\scalarRig{x}{y} = \sum_{a \in \web{X}} x_a y_a$,
\item matrix multiplication
$\_ \cdot \_ : \rig^{\web{Y} \times \web{Z}} \times \rig^{\web{X} \times \web{Y}}
\arrow \rig^{\web X \times \web Z}$ by $(t \cdot s)_{a, c} = \sum_{b \in \web{Y}} s_{a, b} \cdot t_{b, c} $,
\item and the application 
of a matrix to a vector $\_ \cdot \_ : \rig^{\web{X} \times \web{Y}} \times \rig^{\web X} 
\arrow \rig^{\web Y}$ (this is an instance of matrix multiplication) by $(s \cdot x)_b = \sum_{a \in \web X} s_{a,b} x_a$.
\end{itemize}

\section{Web models based on strong PCMs}
\label{sec:positive-web}

\subsection{Orthogonality and objects}

First we assume that $\rigpos$ is a strong PCR, and we consider 
$\ball \subseteq \rigpos$. We assume that 
$\ball$ is \emph{downward closed}:
if $x \leq y$ and $y \in \ball$ then $x \in \ball$.
This assumption is not necessary to build the model, but it will be
necessary to ensure that the construction yields PCM categories.
We define an \emph{orthogonality relation} $x \orthrel y$ on vectors 
and the \emph{orthogonal} of a set
$F \subseteq \vrigpos{X}$ as:
\[ x \orthrel  y \text{ if } \scalarRig{x}{y} \text{ is defined and belongs to } \ball \qquad 
\orth{F} = \set{x' \in \vrigpos X \suchthat x \orthrel x'}. \]
This construction has the usual properties of an orthogonality
construction: $F \subseteq \biorth{F}$, 
$(F \subseteq G) \imply (\orth{G} \subseteq 
\orth{F})$, and $\triorth{F} = \orth{F}$.
%
%
Different choices of $\rigpos$ and $\ball$ yield the usual definitions 
of orthogonality in various web models
(see \cref{ex:web-models-positive}).

\begin{definition} \label{def:covering}
  A set of vectors $F \subseteq \vrigpos{X}$ is a covering
  if for all $a \in \web X$, 
  there exists $x \in F$ such that its $a$ component $x_a$ is invertible.
\end{definition}

\begin{definition} \label{def:rigpos-space}
An $\rigpos$-space
is a pair $X = (\web X, \points X)$ where 
$\web{X}$ is a set and $\points X \subseteq \vrigpos{X}$ is a set of vectors 
such that $\biorth{\points X} = \points X$, and such that 
both $\points X$ and $\orth{\points X}$ are coverings.
The dual of an $\rigpos$-space $X$ 
is defined as the $\rigpos$-space
$\orth{X} = (\web{X}, \orth{\points X})$.
\end{definition}

The downward closure of $\ball$ induces the downward closure of 
$\rigpos$-spaces, for the pointwise preorder on $\vrigpos X$ 
induced by the preorder on $\rigpos$ given in \cref{def:preorder-sum}.

\begin{lemma}[Downward closure] \label{prop:downward-closure}
For all $x, y \in \vrigpos X$
such that $x \leq y$, $y \in \points X \imply x \in \points X$. 
\end{lemma}

\begin{example} \label{ex:web-models-positive}
This definition and the PCRs of \cref{ex:pcr} 
describe many web models. \begin{itemize}
\item Relational model: take $\rigpos$ to be the complete boolean rig $\{0,1\}$ (the infinite 
sum $1+1+\dots$ is defined and equal to $1$), and $\ball = \rigpos$.
\item Weighted relational model \cite{LairdManzonettoMcCuskerPagani13}: take $\rigpos$ to be any 
complete rig (or continuous, if we want fixpoints), and $\ball = \rigpos$.
\item Probabilistic coherence spaces \cite{DanosEhrhard11}: 
take $\rigpos = \Rpos$ (recall \cref{ex:real}) 
and $\ball = [0, 1]$.
\item Coherence spaces \cite{Girard87}: take $\rigpos = \{0, \omega\}$ with $\omega + \omega$ undefined, 
and $\ball = \{0, \omega\}$.
\item Finiteness spaces (relational version) \cite{Ehrhard05}: take $\rigpos$ to be the 
finitary PCM induced by the finite boolean rig $\{0,1\}$ (the infinite 
sum $1+1+\dots$ is not defined), and $\ball = \rigpos$.
\item Köthe spaces (positive version) \cite{Ehrhard02}: take $\rigpos = \Rpos$ and $\ball = \Rpos$.
\end{itemize}

Different choices of PCR describe different computational properties: 
the countable non-determinism of the (weighted) relationnal model, 
the randomness of probabilistic coherence spaces, 
the determinism for coherence spaces, and the finite non-determinism
for finiteness spaces.
There are however some web models that cannot be described 
by an orthogonality relation, for instance the category non-uniform 
coherence spaces \cite{BucciarelliEhrhard01}.
\end{example}

\subsection{Tensor product of objects and morphisms}

\label{sec:positive-web-morphism}

We can now build the category of $\rigpos$-spaces following 
the standard procedure of web based models: we 
first build the tensor product of objects, and we deduce
the morphisms by duality.
For all vectors $x \in \vrigpos{X}$ and $y \in \vrigpos{Y}$,
define $x \sm y \in \rigpos^{\web{X} \times \web{Y}}$  
by $(x \sm y)_{a, b} = x_a y_b$. For all sets $F \subseteq \vrigpos{X}$
and $G \subseteq \vrigpos{Y}$, we set
$F \sm G = \set{x \sm y \suchthat x \in F, y \in G}$.
The \emph{tensor} $X \sm Y$ of two $\rigpos$-spaces
$X$ and $Y$ is defined as
\[ \web{X \sm Y} = \web{X} \times \web{Y} 
\qquad \points{(X \sm Y)} = \biorth{(\points X \sm \points Y)}. \]
The closure under double orthogonality is crucial to ensure that 
this is an $\rigpos$-space.
We now define the $\rigpos$-space
of linear maps between $\rigpos$-spaces by duality:
$X \linarrow Y = \orth{(X \sm \orth{Y})}$.

\begin{remark} \label{rem:linarrow-predual}
  $\points(X \linarrow Y) = \orth{\points(X \sm \orth{Y})} 
= \triorth{(\points X \sm \points\orth{Y})} = \orth{(\points X \sm \points\orth{Y})}$.
Thus, a matrix $s \in \rig^{\web{X \linarrow Y}}$ belongs 
to $\points(X \linarrow Y)$ if and only if 
$s \orthrel (x \sm y')$ for all $x \in \points X, y' \in \points \orth Y$.  
This predual characterization of $X \linarrow Y$ 
is the backbone of web semantics.
\end{remark}

\begin{restatable}{theorem}{compositionwebpos} 
  \label{thm:composition-web-pos}
For all $s \in \points(X \linarrow Y)$ and $t \in \points(Y \linarrow Z)$,
$t \cdot s$ is well-defined and belongs to 
$\points(X \linarrow Z)$.
Furthermore, the identity matrix $\id_X \in \vrigpos{X \linarrow X}$
belongs to $\points{(X \linarrow X)}$. 
Thus, we can define the category $\WEBPOS$ whose objects 
are the $\rigpos$-spaces, whose morphisms
are $\WEBPOS(X, Y) = \points(X \linarrow Y)$, and
whose composition is given by matrix multiplication.
\end{restatable}

\begin{proof}[Proof sketch of \cref{thm:composition-web-pos}]
The proof (see \cref{app:web-models-positive}) uses the predual characterization
described in \cref{rem:linarrow-predual}.
Now, observe that $\scalar{s}{x \sm y'}$ can be reordered as follows.

\begin{restatable}{lemma}{scalarrearanging} \label{prop:scalar-rearanging}
For all $s \in \rigpos^{\web{X} \times \web{Y}}$, $x \in \vrigpos X$ and 
$y \in \vrigpos{Y}$,  \begin{itemize}
  \item $\scalar{s \cdot x}{y} \sumsub \scalar{s}{x \sm y}$. Furthemore, 
  if $s \cdot x$ is defined, then 
  $\scalar{s \cdot x}{y} \sumiff \scalar{s}{x \sm y}$.
  \item $\scalar{\orth{s} \cdot y}{x} \sumsub \scalar{s}{x \sm y}$. 
  Furthemore, if $\orth{s} \cdot y$ is defined, then 
  $\scalar{\orth{s} \cdot y}{x} \sumiff \scalar{s}{x \sm y}$. 
\end{itemize}
where $\orth{s} \in \rigpos^{\web Y \times \web X}$ is the transpose of $s$,
$(\orth s)_{b, a} = s_{a, b}$
\end{restatable}
Unfortunately, the equivalences in \cref{prop:scalar-rearanging} 
requires $s \cdot x$ and $\orth{s} \cdot y$ 
to be defined. Fortunately, we can 
ensure that this condition often holds, 
by \cref{prop:covering-principle-sm} below.

\begin{restatable}[Covering principle]{lemma}{coveringtwo} 
  \label{prop:covering-principle-sm}
  Let $Q$ be a covering of $\vrigpos Y$. Let $x \in \vrigpos X$ and 
  $s \in \vrigpos{X \linarrow Y}$ such that 
  $\scalar{s}{x \sm y}$ is well-defined for all $y \in Q$.
  Then $s \cdot x$ is well-defined.
  
  Similarly, let $P$ be a covering of $\vrigpos X$. Let $y \in \vrigpos Y$ 
  and $s \in \vrigpos{X \linarrow Y}$ such that 
  $\scalar{s}{x \sm y}$ is well-defined for all $x \in P$.
  Then $\orth{s} \cdot y$ is well-defined.
\end{restatable}
Used together, \cref{prop:scalar-rearanging,prop:covering-principle-sm} 
give us a convenient 
characterization of
$\points{(X \linarrow Y)}$.
This characterization gives a unified account 
of a characterization that appears in all web models, 
such as probabilistic coherence spaces  
(Lemma~3 of~\cite{DanosEhrhard11}), finiteness 
spaces (Lemma~3 of~\cite{Ehrhard05}), and even 
the historical model of coherence 
spaces (see Definition~3.6 of~\cite{Girard87}).

\begin{restatable}{proposition}{linarrowcharacterization} 
  \label{prop:linarrow-characterization}
  Let $s \in \vrigpos{X \linarrow Y}$. The following are equivalent: 

  \noindent
  \begin{minipage}{0.3 \linewidth} 
    \begin{enumerate}
      \item $s \in \points{(X \linarrow Y)}$;
      \item $\orth{s} \in \points{(\orth Y \linarrow \orth X)}$;
    \end{enumerate}
  \end{minipage}
  \begin{minipage}{0.6 \linewidth}
    \begin{enumerate} \setcounter{enumi}{2}
      \item $\forall x \in \points X$, $s \cdot x$ is defined and belongs to $\points Y$;
      \item $\forall y' \in \orth{\points Y}$, $\orth{s} \cdot y'$ is defined 
        and belongs to $\points \orth{X}$.
    \end{enumerate}
  \end{minipage}
\end{restatable}
The proof of \cref{thm:composition-web-pos} uses this characterization.
First, we prove that $t \cdot s$ is well-defined with an 
argument similar to that of \cref{prop:covering-principle-sm}.
Then, we prove that $(t \cdot s) \cdot x \sumiff t \cdot (s \cdot x)$,
and we conclude that $t \cdot s \in \points{(X \linarrow Z)}$ 
by applying \cref{prop:linarrow-characterization} twice.
\end{proof}

To wrap up this section, we further improve the predual
characterization of $X \linarrow Y$
given in \cref{rem:linarrow-predual,prop:linarrow-characterization}.
This new predual characterization 
captures in a unified way various proof 
techniques that appear throughout the litterature
on web models.

\begin{restatable}[Predual characterization of morphisms]{theorem}{predualcharacterization} 
  \label{thm:predual-characterization}
  Let $X$ and $Y$ be two $\rigpos$-spaces. Let 
  $P \subseteq \vrigpos{X}$ and $Q \subseteq \vrigpos Y$ be two coverings
  such that $\biorth{P} = \points X$ and $\orth{Q} = \points Y$. 
  Then for all $s \in \vrigpos{X \linarrow Y}$, the following are equivalent.

  \noindent
  \begin{minipage}{0.3 \linewidth} 
    \begin{enumerate}
      \item $s \in \points{(X \linarrow Y)}$;
      \item $s \in \orth{(P \sm Q)}$;
    \end{enumerate}
  \end{minipage}
  \begin{minipage}{0.6 \linewidth}
    \begin{enumerate} \setcounter{enumi}{2}
      \item for all $x \in P$, $s \cdot x$ is defined and belongs to $\points Y$;
      \item for all $y' \in Q$, $\orth s \cdot y'$ is defined and belongs to $\points \orth X$.
    \end{enumerate}
  \end{minipage}
\end{restatable}
The proof of~\cref{thm:predual-characterization} 
relies heavily on \cref{prop:scalar-rearanging,prop:covering-principle-sm}
(see~\cref{app:web-models-positive}). 
The following result immediately follows from \cref{thm:predual-characterization},
using that $X \sm Y = \orth{(X \linarrow \orth{Y})}$.

\begin{corollary} \label{cor:predual-characterization}
  Let $P \subseteq \vrigpos{X}$ and $Q \subseteq \vrigpos Y$ be two coverings
such that $\biorth{P} = \points X$ and $\biorth{Q} = \points Y$.
Then $\points{(X \sm Y)} = \biorth{(P \sm Q)}$.
\end{corollary}

\subsection{The linear logical structure of web models}
\label{sec:positive-web-ll}

We describe here the objects of the linear logical structure
of $\WEBPOS$.
The rest of the structure (symmetric monoidal structure, closure, 
resource comonad, and Seely isomorphisms) is the usual structure of web based models  
(summarized in \cref{fig:web-ll} given in the appendix).
The predual characterization of \cref{thm:predual-characterization}
ensures that all the matrices describing the LL structure (in~\cref{fig:web-ll}) 
are morphisms.
We define the exponential of 
an $\rigpos$-space $X$ by: 
\begin{align*}
\web{\bang X} &= \set{m \in \Mfin{\web X} \suchthat \exists x \in \points X 
\text{ such that } \forall a \in m, x_a \text{ is invertible}}\\
\points{(\bang X)} &= \biorth{\set{\prom x \suchthat x \in \points X}} 
\text{ where } (\prom x)_{[a_1, \ldots, a_n]} = x_{a_1} \cdots x_{a_n} 
\end{align*}
Observe that $\web{\bang X} \neq \Mfin{\web X}$ in general because we want 
$\points{(\bang X)}$ to be a covering. Indeed, in coherence spaces
the web of $\bang X$ is the set of multisets whose support is a clique.

By design, the category is $*$-autonomous, with 
dualizing object $\bot = \orth{\smone}$.
It has products and coproducts. The webs are given by
$\web{\withFam X_i} = \web{\plusFam X_i} = 
\bigcup_{i \in I} \set{i} \times X_i$.
We can define for all $i \in I$ injections 
$\inj_i \in \vrigpos{X_i \linarrow \left(\plusFam X_i\right)}$ 
and projections $\proj_{i} \in \vrigpos{\left(\withFam X_i \right) \linarrow X_i}$
by $(\inj_i)_{a, (j, a')} = (\proj_i)_{(j, a), a'} = \kronecker i j \kronecker a {a'}$.
We can then define 
\[ \points{\Big(\withFam X_i \Big)} = \orth{\set{\inj_i \cdot x_i \suchthat 
x_i \in \points \orth{X_i} }} \qquad 
\points{\Big(\plusFam X_i \Big)} = \biorth{\set{\inj_i \cdot x_i \suchthat 
x_i \in \points X_i }} \]
and we can check that these are products and coproducts.
Observe in particular that for all $x \in \vrigpos{\withFam X_i}$,
$\scalar{\proj_i \cdot x}{x'} \sumiff \scalar{x}{\inj_i \cdot x'}$
so $x \in \points \left(\withFam X_i\right)$ if and only if 
for all $i \in I$, $\proj_i \cdot x \in \points X_i$.
In particular, $\points \left(\plusFam X_i\right) 
\subseteq \points \left(\withFam X_i \right)$, but this is not an equality 
in general. The initial and terminal object is 
$\top = (\emptyset, \set{0})$.

\begin{remark} 
  \label{rem:lafont}
  We conjecture that the exponential we describe here is always 
  Lafont \cite{Lafont88}.
  Except for Köthe spaces, all the web models given in \cref{ex:web-models-positive}
  are known to be Lafont, but the proof changes drastically depending 
  on the model: the proof that 
  probabilistic coherence spaces are Lafont 
  \cite{CrubilleEhrhardPaganiTasson17} uses a generic 
  formula  that fails 
  in finiteness spaces~\cite{MelliesTabareauTasson09}. 
\end{remark}

\section{Web models based on absolute convergence}

\label{sec:negative-web}

We now describe web models based on PCRs that are not strong, but whose 
summability behaves similarly to the absolute convergence 
on real numbers (\cref{ex:real}). We give two examples:
Köthe spaces \cite{Ehrhard02}, and the 
category of modules over a finiteness space \cite{Ehrhard05}.

\begin{definition} \label{def:absolute-pcr}
A PCR $\rig$ is absolute
if there exists a strong PCR $\rigpos$ and a function 
$\abs . : \rig \arrow \rigpos$ (called the absolute value) such that \begin{itemize}
\item $\abs{x} = 0$ if and only if $x = 0$;
\item $\abs{.}$ is a monoid morphism for the multiplicative part: 
$\abs{x y} = \abs{x} \abs{y}$ and $\abs{1} = 1$;
\item  Triangle inequality: $\family{x_a}$ is summable in $\rig$ if and 
only if $\family{\abs{x_a}}$ is summable in $\rigpos$, and 
$\abs{\sum_{a \in A} x_a} \leq \sum_{a \in A} \abs{x_a}$.
\end{itemize}
\end{definition}

We generalize to inequalities the notations of \cref{notation:sum-def}.
For example, the triangle inequality stated in \cref{def:absolute-pcr} 
above is described by the equation
$\abs{\sum_{a \in A} x_a} \leqsumiff \sum_{a \in A} \abs{x_a}$.

\begin{example}
The PCR of real numbers is absolute, taking $\rigpos = \Rpos$ and 
$\abs{.}$ to be the regular absolute value.
The finitary PCR induced by a rig 
(\cref{ex:pcr}) is also absolute, taking
$\rigpos$ to be the finitary PCR induced by the boolean rig, and 
$\abs{x} = \kronecker{x}{1}$.
\end{example}

The absolute value $\abs{.}$ extends to a function 
$\posfunction : \vrig{X} \arrow \vrigpos{X}$ for all web $\web{X}$,
defined by $\pos{x}_a = \abs{x_a}$ for all $a \in \web{X}$.
By triangle inequality, for all
$s \in \vrig{X \linarrow Y}$ and $t \in \vrig{Y \linarrow Z}$,
and for all $x, y \in \vrig X$, 
$\pos{t \cdot s} \leqsumiff \pos{t} \cdot \pos{s}$ and
$\abs{\scalar{x}{y}} \leqsumiff \scalar{\pos{x}}{\pos{y}}
$.

\begin{definition} \label{def:module-rig}
For all $\rigpos$-space $X$, let 
$\semimod X = \set{x \in \vrig X \suchthat \pos{x} \in \points X}$.
We define the category $\WEB$ whose objects are the $\rigpos$-spaces, 
whose morphisms from $X$ to $Y$ are the matrices 
in $\semimod{X \linarrow Y}$, and whose composition is matrix multiplication.
\end{definition}
For all $s \in \semimod {X \linarrow Y}$ and 
$t \in \semimod{Y \linarrow Z}$, 
we have that 
$\pos{t \cdot s} \leqsumiff \pos{t} \cdot \pos{s} \in \points{(X \linarrow Z)}$
so $t \cdot s$ is well-defined and belongs to 
$\semimod{X \linarrow Z}$ by downward closure 
of $\points{(X \linarrow Z)}$ (\cref{prop:downward-closure}).
We can check that $\WEB$ is a model of linear logic, with the 
same structure as the one described in \cref{sec:positive-web-ll}.
These matrices are all morphisms in $\WEB$, 
because $\pos{x \sm y} = \pos{x} \sm \pos{y}$,
$\pos{\prom x} = \prom{\pos x}$ and $\pos{\bang s} \leqsumsubinv \bang \pos{s}$.

\begin{example} 
  \textbf{Finiteness spaces.} The category of modules over a finiteness spaces given in section 
4 of \cite{Ehrhard05} is an instance of our \cref{def:module-rig}
in which $\rig$ is a finitary PCR.

\textbf{Köthe spaces.} Köthe spaces are defined in \cite{Ehrhard02} 
a bit differently, by a generalization of the orthogonality 
relation given in \cref{sec:positive-web} and of the 
notion of $\rigpos$-space given in \cref{def:rigpos-space} to the absolute 
PCR of real numbers. 
In the end though, by Proposition 2.12 of~\cite{Ehrhard02} 
the category of Köthe spaces coincides with our definition 
of $\WEB[\R]$.
In fact, we can check that a  
Köthe space $X = (\web X, \kothepoints X)$ in the sense of 
Definition 2.8 of \cite{Ehrhard02} is exactly the 
same as the module $\semimod[\R]{X'}$ where $X'$ is the 
$\Rpos$-space given by $X' = (\web X, \set{\pos{x} \suchthat x \in \kothepoints X})$.

\end{example}

\section{Summability in web models}

\label{sec:summability-web}

We prove in this section 
that both the category $\WEBPOS$ defined in \cref{sec:positive-web}
and the category $\WEB$ defined in \cref{sec:negative-web} have 
a PCM structure on their morphisms.
In fact, we prove that for all $\rigpos$-space $X$,
both $\points X$ and $\semimod{X}$ have a PCM structure.
We start with $\points X$.

\begin{definition} \label{def:summability-web-positive}
Let $X$ be an $\rigpos$-space. 
A family $(x(i))_{i \in I}$ of elements of $\points X$ is summable
 if their pointwise sum $\left(\sum_{i \in I} x(i)\right)_a = 
 \sum_{i \in I} x(i)_a$ 
 is defined and belong to $\points X$.
\end{definition}

\begin{proposition}
  Let $X$ be an $\rigpos$-space. Then $\points X$ is a strong PCM.
\end{proposition}

\begin{proof}
  It immediately follows from \ref{ax:pa} on $\rigpos$ that for all 
  family $(x(i))_{i \in I}$ of $\vrigpos X$ 
  and all partition $\set{I_k}_{k \in K}$ of $I$,
  $\sum_{i \in I} x(i) \sumiff 
  \sum_{k \in K} \sum_{i \in I_k} x(i)$.
  To conclude, we still need to ensure that whenever 
  $(x(i))_{i \in I}$ is summable in $\points X$, then 
  for all $k \in K$,
  $(x(i))_{i \in I_k}$ is summable in $\points X$.
  This is an immediate consequence of the downward closure of 
  $\points X$ (\cref{prop:downward-closure}).
\end{proof}

\begin{remark} \label{rem:semimodules}
  If the set $\ball$ used to define orthogonality 
  is stable under multiplication, then 
  $\ball$ is a strong PCR and  
  $\rigpos$-spaces are $\ball$-modules
  with orthogonal basis \cite{TsukadaKazuyuki22}.
\end{remark}

  Similarly to morphisms, summability can be characterized in terms
  of predual. 
  The proof is very similar to the proof of 
  \cref{prop:covering-principle-sm} and of 
  \cref{thm:predual-characterization}.

  \begin{proposition} \label{prop:summability-predual}
    Let $P$ be a covering of $\vrigpos X$ such that 
    $\orth{P} = \points X$. Then a family
    $(x(i))_{i \in I}$ of $\points X$ is summable if and only if 
    for all $x' \in P$, $\scalar{\sum_{i \in I} x(i)}{x'}$ is defined and 
    is in $\ball$.  
  \end{proposition}

We now define the summability in $\semimod{X}$.

\begin{definition} \label{def:summability-web}
  Let $X$ be an $\rigpos$-space.
  A family $(x(i))_{i \in I}$ of elements of $\semimod X$ is summable
  if $(\pos{x(i)})_{i \in I}$ is summable in $\points X$ (in the 
  sense of \cref{def:summability-web-positive}). 
\end{definition}
By triangle inequality, $\pos{\sum_{i \in I} x(i)} \leqsumiff 
\sum_{i \in I} \pos{x(i)} \in \points X$
so the pointwise sum $\sum_{i \in I} x(i)$ is well-defined 
and belongs to $\semimod{X}$ by downward closure
of $\points{X}$.

  \begin{proposition}
  The pointwise sum defined above is a PCM on $\semimod X$. 
  \end{proposition}

  \begin{proof}
  Let $(x(i))_{i \in I}$ be summable in $\semimod X$, and
  $\set{I_k}_{k \in K}$ be a partition of $A$.
  For all $k \in K$, $(\pos{x(i)})_{i \in I_k}$ is summable in $\points X$ 
  by \cref{prop:subfamily-summable}. Thus, 
  $(x(i))_{i \in I_k}$ is summable in $\semimod X$: 
  \[ \sum_{k \in K} \pos{\sum_{i \in I_k} x(i)}
  \eqnum{1}{\leqsumsubinv} \sum_{k \in K} \sum_{i \in I_k} \pos{x(i)} 
  \eqnum{2}{\sumiff} \sum_{a \in A} \pos{x_a} \in \points X \] 
 where  \texteqnum{1} holds by triangle inequality and  \texteqnum{2} by \ref{ax:pa} on 
  the strong PCM $\points X$.
  We conclude that $(\sum_{i \in I_k} x(i))_{k \in K}$ 
  is summable in $\semimod{X}$ by downward closure of $\points X$.
  \end{proof}

\part{Taylor expansion}

\label{part:taylor}
We generalize coherent Taylor expansion~\cite{EhrhardWalch25} to models 
of LL enriched over \emph{non-necessary strong} PCMs (\cref{def:pcm-category}),
%
that subsume the monoids of
differential categories~\cite{BluteCockettSeely06}.
First, let us give a bird's-eye view of the theory developed 
in \cite{EhrhardWalch25}. 
A \emph{summability structure} on $\catLL$ is, in essence, a functor 
$\S$ with projections $\Sproj_i : \S \naturalTrans \idfun$
such that 
$\sequence{f_i : X \arrow Y}$ is summable 
if and only if there exists 
$f : X \arrow \S Y$ such that $\Sproj_i \compl f = f_i$.

\begin{remark} \label{rem:ss-product-coproduct}
  Defining summability by an endofunctor is 
  standard in algebraic programming 
  semantics. 
  However, the object $\S X$ is often a coproduct, as in 
  partially additive categories~\cite{ManesArbib86}.
  Summability structures are more general, and enjoy an epi-mono factorization (see~\cref{app:ss-coproduct-product}):
  $\begin{tikzcd}
        {\bigoplus_{i \in \N} X } & {\S X} & {\prod_{i \in \N} X}
        \arrow["e", from=1-1, to=1-2]
        \arrow["{m}", from=1-2, to=1-3]
    \end{tikzcd}$.
  %
  When sums are total, this factorization collapses 
  into an equality, and thus the category
  has \emph{countable biproducts}.
\end{remark}

The axioms of (strong) PCMs then endow the summability structure
$\S$ with a canonical \emph{bimonad} 
structure (a monad 
and a comonad structure that interact nicely)~\cite{MesablishviliWisbauer11}.
Furthermore, the compatibility between the PCM structure of $\catLL$ 
and its LL structure boils down to categorical structures 
on the monad $\S$: a lax monoidal structure (wrt $\sm$), 
a strong monoidal structure (wrt $\with$), and 
an invertible pointwise structure 
$\S(X \linarrow Y) \arrow (X \linarrow \S Y)$~\cite{Kock71}.

 In quantitative semantics, morphisms  in the Kleisli category
of the resource comonad are seen as analytic maps.
 Taylor expansion is thus described
as a functor $\T$ on $\klexp$~\cite{EhrhardWalch25}. The action on object
$\T X = \S X$ ensures that all  sums in the Taylor 
expansion are well-defined. The Taylor expansion 
$\T f \in \klexp(\S X, \S Y)$ of
$f \in \klexp(X, Y)$ is intuitively the function:
\begin{equation} \label{eq:coherent-Tinf-explicit}
    \T f : \sequence{x_i} \mapsto 
    \bigg(\sum_{1\leq k\leq j} 
    \ \sum_{i_1 + \cdots + i_k = j}
\frac{1}{\factorial k} \hod{f}{k}(x_0)
\cdot (x_{i_1}, \ldots, x_{i_k}) \bigg)_{j \in \N}. 
\end{equation}
The $j$-th coefficient of $\T f$ corresponds to 
the degree $j$ component of the Taylor expansion of 
$f\left(\sum_{i \in \N} x_i \formalvar^i\right)_{j \in \N}$. 
In particular, $\T f (x, u, 0, \ldots) = \left(\frac{1}{\factorial j} \hod f j 
(x) \cdot (u, \ldots, u) \right)$, so $\T f$ captures the Taylor expansion 
of $f$.
The intuitive~\cref{eq:coherent-Tinf-explicit} is validated by
web model examples in~\cite{EhrhardWalch25}. Furthermore,  this equation is used in \cite{Walch25} to build
$\T$ in cartesian differential 
categories~\cite{BluteCockettSeely09}.

Now, observe that for all $f \in \catLL(X, Y)$, $\Der f \in \catLL(\bang X, Y)$
is a linear map, so we should have that 
$\T(\Der f)(x_i)_{i \in \N} = \left(\Der f(x_i) \right)_{i \in \N}$,
thus $\T(\Der f) = \Der(\S f)$. We say that 
$\T$ is an \emph{extension} of $\S$ to $\klexp$ \cite{PowerWatabane02}.
Such extension boils down to  
a \emph{distributive law} \cite{PowerWatabane02} between 
the functor $\S$ and the comonad $\bangtext$. It consists 
of a natural transformation 
$\Sdl : \bang \S \naturalTrans \S \bang$ subject to a compatibility
condition with $\der$ and $\dig$.
Additionnal compatibility 
conditions between $\Sdl$ and the structure on 
$\S$ (bimonad structure, monoidal structure with regard to 
$\sm$ and $\with$, and projection $\Sproj_0$) ensure that the structure 
on $\S$ extends to  $\T$. This structure on $\T$ describes with 
categorical equations the standard properties of the differential 
calculus~\cite{Walch25}.

Summability is often given by $\S = \Dbimon 
\linarrow \_$, with $\Dbimon = \withFam 1$ (see~\cite{EhrhardWalch25}).
The object $\Dbimon$ has a  bimonoid structure
that is \emph{mate} \cite{KellyStreet74} to 
the bimonad structure of $\S$ (see \cref{fig:mate}). 
The mate isomorphism implies that the existence of the distributive law 
$\Sdl : \bang \S \naturalTrans \S \bang$ is equivalent to the existence 
of a $\bang$-coalgebra $\Dbimonca \in \catLL(\Dbimon, \bang \Dbimon)$ 
compatible with the bimonoid structure of $\Dbimon$.

\paragraph{Contribution.}
We adapt~\cite{EhrhardWalch25} to a setting 
where PCMs are not necessarily strong. 
In~\cref{sec:taylor-pcm}, we present summability structures in the absence of~\ref{ax:pa}.
In~\cref{sec:representable}, 
we replace the object $\Dbimon = \withFam 1$ as it does not
account for respresentable summability in Köthe spaces.
In~\cref{sec:web-model-representable}, we show 
that all web models defined in \cref{part:web-models} can be captured  
by this theory.

\section{Taylor expansion in PCM categories}

\label{sec:taylor-pcm}

\subparagraph*{Summability structure.} We generalize
the summability structure of \cite{EhrhardWalch25}
to PCM categories.


\begin{definition} \label{def:ss} \labeltext{(S-fun)}{def:ss-fun} 
    \labeltext{(S-proj)}{def:ss-proj} 
    \labeltext{(S-sum)}{def:ss-sum}
    \labeltext{(S-sum-1)}{def:ss-sum-1}
    \labeltext{(S-sum-2)}{def:ss-sum-2}
    \labeltext{(S-sum-3)}{def:ss-sum-3}   
A \emph{summability structure} on a category $\catLL$
is a tuple $(\S, \sequence{\Sproj_i}, \Ssum, 0)$ 
where \begin{itemize}
    \item \ref{def:ss-fun}: $\S$ is a functor on $\catLL$, and
    $\Sproj_i, \Ssum :  \S \naturalTrans \idfun$
    are natural transformations called respectively projections and sum, and $0$ 
    is a family of zero morphisms $0^{X, Y} \in \catLL(X, Y)$. 
    \item \ref{def:ss-proj} The $\Sproj_i$ are jointly monic: 
    $\forall f, g \in \catLL(X, \S Y)$, 
    $(\forall i \in \N, \Sproj_i \compl f = \Sproj_i \compl g) \implies f = g$.
    \item \ref{def:ss-sum} 
    For all objects $X, Y$, there exists a PCM structure on 
    $\catLL(X, Y)$ such that \begin{itemize} 
        \item \ref{def:ss-sum-1} 
        A family $(f_{i_1, \ldots, i_n})_{i_1, \ldots, i_n \in \N}$ over $\catLL(X, Y)$
    is summable if and only if there exists $f \in \catLL(X, \S^n Y)$
    such that $\Sproj_{i_1} \compl \cdots \compl \Sproj_{i_n} \compl f = 
    f_{i_1, \ldots, i_n}$. 
    \item \vskip-0.7em \ref{def:ss-sum-2}
    If $(f_{i_1, \ldots, i_n})_{i_1, \ldots, i_n \in \N}$ is summable then 
    $\sum_{i_1, \ldots, i_n \in \N} f_{i_1, \ldots, i_n} = \overbrace{\Ssum \compl \cdots \compl 
    \Ssum}^n \compl f$.
    \item \ref{def:ss-sum-3} The sum over the empty family is equal to the zero morphism $0$.
    \end{itemize}
\end{itemize} 
\end{definition}
The axiom \ref{def:ss-proj} ensures that the object $\S X$
describes $\N$-indexed families over $X$. It is similar to a
product, but the pairing only exists if the family 
is summable by \ref{def:ss-sum}.

\begin{definition}
By \ref{def:ss-proj}, the morphism $f$ given in \ref{def:ss-sum} 
is unique and is called the witness of 
$(f_{i_1, \ldots, i_n})_{i_1, \ldots, i_n \in \N}$.
We write this witness 
$\Spairing[(i_1, \ldots, i_n)]<\N^n>{f_{i_1, \ldots, i_n}}$.
\end{definition}


The reader may be puzzled as for why we consider $\N^n$ indexed 
families in \ref{def:ss-fun} instead of simply $\N$.
We need to consider these families to ensure that the iteration of $\S$ 
does not introduce more constraints than 
the summability of each components.

\begin{example}
In web models, the summability structure is given by 
$\web{\S X} = \bigcup_{i \in \N} \set{i} \times \web{X}$ and
$\points{(\S X)} = \set{\sequence{x_i} \in \left(\vrigpos X\right)^{\N} \suchthat 
\sequence{x_i} \text{ is summable}}$,
with projections $\Sproj_i \cdot \sequence{x_i} = x_i$ and sum  
$\Ssum \cdot \sequence{x_i} = \sum_{i \in \N} x_i$.
Observe that $\S X$ sits between the product and the coproduct, as 
mentionned in \cref{rem:ss-product-coproduct}.
We will prove in \cref{sec:web-model-representable} that our web models
are \emph{representable}, which will ensure \emph{de facto}
that $\S$ is a summability structure.  
\end{example}

We do not assume in \cref{def:ss} above that $\catLL$ is a PCM category.
This is captured by
\ref{def:ss-fun}, as shown in \cref{prop:ss-is-pcm-cat} below.

\begin{proposition} \label{prop:ss-is-pcm-cat}
If $\catLL$ is equipped by a summability structure, then $\catLL$ is a PCM category.
\end{proposition}

\begin{proof}
    We need to prove strong distributivity. Let $\vect f = (f_a)_{a \in A}$ be a summable 
    family of $\catLL(X, Y)$, and $\vect g = (g_a)_{a \in A}$ be a summable 
    family of $\catLL(Y, Z)$. 
    Let $\phi_1 : A \injection \N$
    and $\phi_2 : B \injection \N$ be any injections (they exist because 
    $A$ and $B$ are countable).
    Let $(f'_{i})_{i \in \N} = \famact {\phi_1} {\vect f}$ and 
    $(g'_{j})_{j \in \N} = \famact {\phi_2} {\vect g}$.
    These families are summable by assumption and \cref{prop:reindexing}.

    Let $f' \in \catLL(X, \S Y)$ and 
    $g' \in \catLL(Y, \S Z)$ be the witnesses
    given by \ref{def:ss-sum-1}.
    Then, by naturality of the $\Sproj_i$,
    $\Sproj_j \compl \Sproj_i \compl \S g' \compl f' 
    = \Sproj_j \compl g' \compl \Sproj_i \compl f' 
    = g'_j \compl f'_i$.
    Thus, by \ref{def:ss-sum-1},
    $(g'_j \compl f'_i)_{i, j \in \N}$
    is summable with sum 
    $\Ssum \compl \Ssum \compl \S g' \compl f' \eqnum{1}= 
    \Ssum \compl g' \compl \Ssum \compl f' \eqnum{2}= 
    \left(\sum_{j \in \N} g'_j \right) \compl 
    \left(\sum_{i \in \N} f'_i \right) 
    \eqnum{3}= \left( \sum_{b \in B} g_b \right) \compl 
    \left(\sum_{a \in A} f_a \right)$
    where \texteqnum{1} is naturality of $\Ssum$,
    \texteqnum{2} is \ref{def:ss-sum-2} and \texteqnum{3} is \cref{prop:reindexing}.
    To conclude the proof that
    $\catLL$ is a PCM category, we only need by \cref{prop:reindexing}
    to give an injection $\psi : A \times B \arrow \N^2$
    such that 
    $(g'_j \compl 
    f'_i)_{i, j \in \N} 
    = \famact \psi {(h_{a, b})_{(a, b) \in A \times B}}$ with  
    $h_{a, b} = g_b \compl f_a$.
    Define $\psi$ as $\psi(a, b) = (\phi_1(a), \phi_2(b))$.
    Then, for all $i, j \in \N$,
    \begin{itemize}
        \item if $(i, j) \in \im(\psi)$,  then 
        $i \in \im(\phi_1)$ and $j \in \im(\phi_2)$, and 
        $g'_j \compl f_i' = g_{\phi_2^{-1}(j)} \compl f_{\phi_1^{-1}(i)} 
        = h_{\psi^{-1}(i, j)}$;
        \item otherwise, either $i \notin \im(\phi_1)$ and $f_i = 0$,
        or $j \notin \im(\phi_2)$ and $g_j = 0$.
        By \ref{def:ss-sum-3}, $0$ is a zero morphism 
        and thus $g_j \compl f_i = 0$. \qedhere
    \end{itemize}
\end{proof}

As shown in Section 3.4 of \cite{EhrhardWalch25}, every summability structure 
$\S$ has a canonical \emph{bimonad} structure \cite{MesablishviliWisbauer11}. 
%
The proof of \cite{EhrhardWalch25} carries to our setting, 
so we simply recall the result here.



\begin{proposition}\label{prop:bimonad-structure}
  \begin{enumerate}
    \item For all $i \in \N$, 
    there exists 
    $\Sinj_i : \idfun \naturalTrans \S$
    characterized by \(\Sproj_j\comp\Sinj_i=\kronecker ij  \).
    \item There exists a natural transformation $\SmonadSum : \S \S \naturalTrans \S$
  characterized by
  \(\Sproj_i\comp\SmonadSum =\sum_{j=0}^i\Sproj_{i-j}\comp\Sproj_j\).
    \item There exists a natural transformation 
    $\Slift : \S \naturalTrans \S \S$ characterized
    by $\Sproj_i \comp \Sproj_j \comp \Slift = \kronecker ij \Sproj_i$.
    \item There exists a natural transformation 
    \(\Sswap : \S \S \naturalTrans \S \S \) 
    characterized by
  \(\Sproj_i\comp\Sproj_j\comp\Sswap=\Sproj_j\comp\Sproj_i\).
  \end{enumerate}
\end{proposition}


\begin{theorem}[Theorem 74 of \cite{EhrhardWalch25}]
$(\S, \Sinj_0, \SmonadSum)$ is a monad, $(\S, \Ssum, \Slift)$ is a comonad, 
$\Sswap : \S \S \naturalTrans \S \S$ is a distributive law 
between the two~\cite{Beck69,PowerWatabane02}, and 
together they form a bimonad.
\end{theorem}

By monicity of the $\Sproj_i$, a morphism $f \in \catLL(X, \S Y)$ 
of the Kleisli category of the monad $\S$ can be seen as a formal power 
series $\sum_{i \in \N} f_i \formalvar^i$ over $\catLL(X, Y)$, 
where $f_i = \Sproj_i \compl f$ and 
$\formalvar$ is a formal variable. 
Then the composition of two power series in the Kleisli category of $\S$
is given by their Cauchy product, 
$\left(\sum_{j \in \N} g_j \formalvar^j \right) \compl 
\left(\sum_{i \in \N} f_i \formalvar^i \right)
= \sum_{k \in \N} \left( \sum_{i + j = k} g_j \compl f_i\right) \formalvar^k$.
As such, the Kleisli category of $\S$ is quite reminiscent of the monoid semiring 
construction of~\cite{Hines13}.

\subparagraph*{Compatibility between summability and categorical structure.}
We detail how the PCM structure of a PCM category $\catLL$
should interact with the symmetric monoidal (closed) structure of 
$\catLL$, and its products, whenever they exist.

\begin{definition} \label{def:sum-sm}
The PCM structure is compatible with the symmetric monoidal product 
if $\sm$ is strongly distributive: for all indexed families 
$\family{f_a : X_1 \arrow Y_1}$ and $\family<B>[b]{g_b : X_2 \arrow Y_2}$,
$\left(\sum_{a \in A} f_a \right) \sm \left(\sum_{b \in B} g_b\right) 
\sumsub \sum_{a \in A, b \in B} (f_a \sm g_b)$.
\end{definition}

\begin{definition} The PCM structure is compatible with the monoidal closure 
if it is compatible with the symmetric monoidal product and 
if for all indexed family $\family{f_a : X \sm Y 
	\arrow Z}$, 
	$\sum_{a \in A} \cur(f_a) 
	\sumiff \cur \left(\sum_{a \in A} f_a \right)$.
\end{definition} 

\begin{definition}
The PCM structure is compatible with the finite products (respectively countable products)
  if for all finite (respectively countable) index set $I$
  and for all collection of indexed families 
  $\family{f_a^i \in \catLL(X, Y_i)}$ where $i$ ranges over $I$,
  $\sum_{a \in A} \prodPairing{f_a^i} \sumiff \prodPairing{\sum f_a^i}$
 \end{definition}

 \begin{remark} \label{rem:sum-ll-weaker}
 These definitions are the most intuitive, but are quite redundant.
 \begin{itemize}
 \item \emph{Monoidal product}: the left and right distributivity of $\sm$ 
 implies strong distributivity, by strong distributivity of the composition 
 in PCM categories. Furthermore, left distributivity 
 implies right distributivity by symmetry of the product, and vice versa.
 \item \emph{Cartesian product}: 
 by left distributivity, $\proj_i \compl \sum_{a \in A} \prodPairing{f_a^i} 
\sumsub \sum_{a \in A} (\proj_i \compl \prodPairing{f_a^i})
= \sum_{a \in A} f_a^i$
so by uniqueness of the pairing, we always have 
$\sum_{a \in A} \prodPairing{f_a^i} 
\sumsub \prodPairing{\sum f_a^i}$.
\item \emph{Monoidal closed structure}: similarly, if the PCM structure is compatible with 
 $\sm$, we can prove that $\sum_{a \in A} \cur(f_a) \sumsub \cur \left(\sum_{a \in A} f_a \right)$.
 \end{itemize}
\end{remark}

If the PCM category structure of $\catLL$
is induced by a summability structure $\S$, 
all the compatibility conditions 
above boil down to categorical structure on $\S$ (see proof 
in~\cref{app:compatibility-pcm-ll}).

\begin{restatable}{proposition}{compatibilitypcmll}
  \label{prop:ss-ll}
  Assume that $\catLL$ is equipped with a summability structure $\S$. 
  \begin{enumerate}
  \item The PCM structure induced by $\S$ is compatible with the symmetric monoidal structure of 
  $\catLL$ if and only if 
  for all objects $X$ and $Y$, $0 \tensor Y = 0$ and there exists
  \begin{equation} \label{eq:sm-strength}
    \begin{split}
    \SstrR_{X, Y} = \Spairing{\Sproj_i \sm Y} 
    : \S X \sm Y \arrow \S (X \sm Y) \text{ such that } 
    \Ssum \compl \SstrL = (\Ssum \sm Y) \\ 
    \SstrL_{X, Y} = \Spairing{X \sm \Sproj_i} 
    : X \sm \S Y \arrow \S (X \sm Y) \text{ such that } 
    \Ssum \compl \SstrR = (X \sm \Ssum)
    \end{split}
  \end{equation}
  \item The PCM structure is compatible with the symmetric monoidal closed structure 
  (whenever there is one) if and only if it is compatible with the 
  symmetric monoidal structure, and if for all object $X$, 
	$\sequence{X \linarrow \Sproj_i} : (X \linarrow \S Y) \arrow 
  (X \linarrow Y)$ is summable.
  \item The PCM structure is compatible with the finite (resp. countable) products 
  (whenever they exist)
  if and only if for all finite (resp. countable) index set $I$, 
  $(\withFam \Sproj_j)_{j \in \N}$ is summable.
  \end{enumerate}
\end{restatable}

The natural transformations 
$\SstrR$ and $\SstrL$ endow $\S$ with the structure 
of a \emph{commutative monad}~\cite{EhrhardWalch25}. This induces
by standard results~\cite{Kock70,Kock72}
a \emph{lax monoidal monad} structure
$\Sinj_0 : 1 \arrow \S 1$ and 
$\Sdist_{X, Y} : \S X \sm \S Y \naturalTrans \S(X \sm Y)$
characterized by $\Sproj_i \compl \Sdist = 
\sum_{i_1 + i_2 = i} \Sproj_{i_1} \sm \Sproj_{i_2}$. 
Similarly, $\S$ has a canonical 
oplax monoidal structure $\prodPair{\S \proj_1}{\S \proj_2} :
\S(X \with Y) \arrow \S X \with \S Y$ 
(Proposition 22 of \cite{Mellies09})
with inverse $\Spairing{\Sproj_i \with \Sproj_i}
: \S X \with \S Y \naturalTrans \S(X \with Y)$.
Thus, the compatibility 
between sums and products turns $\S$ into a strong monoidal 
monad (wrt $\with$).

\subparagraph*{Taylor expansion.}
Our Taylor expansion is the same as 
in \cite{EhrhardWalch25}.

\begin{definition}[Definition 112 of \cite{EhrhardWalch25}] 
  \label{def:taylor-expansion} An analytic Taylor expansion is  a distributive law 
$\Sdl : \bang \S \naturalTrans \S \bang$ between 
the functor $\S$ and the comonad $\bangtext$
that is also compatible with the structure on $\S$
(bimonad, monoidal structure with regard to $\sm$ and $\with$, and 
projection $\Sproj_0$).
\end{definition}

\section{The representable theory}

\label{sec:representable}

We assume in this section that $\catLL$ is a symmetric monoidal 
closed category, with zero morphisms $0$ that are also absorbing for the monoidal product : $X \sm 0 = 0$ and 
$0 \sm Y = 0$.

\begin{definition}
A summation object in $\catLL$ 
is an object $\Dbimon$ equipped with 
a family of jointly epic morphisms 
$(\Dinj_i : \Sone \arrow \Dbimon)_{i \in \N}$
called the basis, 
and $\Ddiag : \Sone \arrow \Dbimon$ called the diagonal.
\end{definition}

In \cite{EhrhardWalch25}, $\Dbimon$ is always assumed to be equal to the product
$\withFam<\N> 1$. However, it turns out that this assumption 
does not give the expected notion of sum in Köthe spaces.
We refer the reader to \cref{sec:web-model-representable} for a generic description
of this summation object in web models.

The internal hom 
\((\Sone\Limpl X,\Evlin_{1, X})\) can be described as 
  \(\Sone\Limpl X=X\), \(\Evlin_{1, X}=\smunitR_X : X \sm 1 \arrow X\) and 
  for all $f : X \sm \smone \arrow Y$,
  $\cur(f) = f \compl \tensorUnitR^{-1}$.
  We define natural transformations
  $\Sproj_i=(\Dinj_i \Limpl X) : (\Dbimon \linarrow X) \arrow X $
  and $\Ssum = (\Ddiag\Limpl X) : (\Dbimon \linarrow X) \arrow X$. More explicitely,
\begin{equation} \label{eq:projection-representable-def}
    \begin{tikzcd}[column sep = large, row sep = abysmal]
      {\Sproj_i = \quad \Dbimon \linarrow X} & {(\Dbimon \linarrow X) \tensor 1} & {(\Dbimon \linarrow X) \tensor \Dbimon} & X. \\
      {\Ssum =  \quad \Dbimon \linarrow X} & {(\Dbimon \linarrow X) \tensor 1} & {(\Dbimon \linarrow X) \tensor \Dbimon} & X.
      \arrow["{\tensorUnitR^{-1}}", from=1-1, to=1-2]
      \arrow["{(\Dbimon \linarrow X) \tensor \Dinj_i}", from=1-2, to=1-3]
      \arrow["{\ev_{\Dbimon, X}}", from=1-3, to=1-4]
      \arrow["{\tensorUnitR^{-1}}", from=2-1, to=2-2]
      \arrow["{(\Dbimon \linarrow X) \tensor \Ddiag}", from=2-2, to=2-3]
      \arrow["{\ev_{\Dbimon, X}}", from=2-3, to=2-4]
    \end{tikzcd}
    \end{equation}

\begin{definition} \label{def:representable-pcm}
  A symmetric monoidal closed category with zero morphisms 
  is representable if it has a summation 
  object $\Dbimon$ such that 
  $(\Dbimon \linarrow \_, \sequence{\Sproj_i}, \Ssum, 0)$ 
  is a summability structure.
\end{definition}

We now provide an alternative characterization of representable 
PCM categories, which can be obtained by ``uncurrying'' the 
definition of summability structures.
By standard naturality equations on the currying, we get that 
for all $f : X \sm \Dbimon \arrow Y$ and $g : X \sm \Dbimon^{\sm n} 
\arrow Y$,
  \begin{align}
     \Sproj_i \compl \cur(f) 
     = f \compl (X \tensor \Dinj_i) \compl \tensorUnitR^{-1}
     && \ \Sproj_{i_1} \compl \cdots \compl \Sproj_{i_n} \compl \cur^{n}(g) 
     &= g \compl (X \sm \Dinj_{i_n} \sm \cdots \compl \Dinj_{i_1}) \compl
     (\tensorUnitR^n)^{-1} \label{eq:curry-proj} \\
     \Ssum \compl \cur(f) 
     = f \compl (X \tensor \Ddiag) \compl \tensorUnitR^{-1}
     && \ \Ssum \compl \cdots \compl \Ssum \compl \cur^{n}(g) 
     &= g \compl (X \sm \Ddiag \sm \cdots \compl \Ddiag) \compl
     (\tensorUnitR^n)^{-1}  \label{eq:curry-sum}
   \end{align}

\begin{lemma} \label{prop:representable-monicity}
  The following assertions are equivalent: \begin{enumerate}
    \item the morphisms $\Dinj_i$ are jointly epic;
    \item for all object $X$, the morphisms 
    $X \sm \Dinj_i : X \sm \smone \arrow X \sm \Dbimon$
    are jointly epic;
    \item the morphisms $\Sproj_i$ are jointly monic.
  \end{enumerate}
\end{lemma}

  \begin{proof}
    $(1) \Rightarrow (2)$. Let \(f,g : \Tens X \Dbimon \arrow Y\) be such that %
    $f\compl\Tensp X{\Dinj_i} =g\compl\Tensp X{\Dinj_i}$ for all $i \in \N$.
    By naturality of \(\Sym\), we get %
    \(f\compl\Sym\compl\Tensp {\Dinj_i}X
    =g\compl\Sym\compl\Tensp {\Dinj_i}X \) and hence %
    \(\Curlin(f\compl\Sym)\compl\Dinj_i
    =\Curlin(g \compl\Sym)\compl\Dinj_i\) %
    so that \(\Curlin(f\compl\Sym)=\Curlin(g\compl\Sym)\) and hence %
    \(f=g\). So the $X \sm \Dinj_i$ are jointly epic. 
    Conversely, $(2) \Rightarrow (1)$, taking $X = \smone$. 
    Finally, $(2) \Leftrightarrow (3)$ by \cref{eq:curry-proj} and 
    bijectivity of $\cur$.
  \end{proof}


  
  Thus, the only condition that does not hold de facto when there is a 
  summation object is \ref{def:ss-sum}. This condition can be uncurryfied 
  by \cref{eq:curry-proj,eq:curry-sum} as follows.

  \begin{proposition} \label{prop:representable-uncurry}
    The SMCC with zero morphisms $\catLL$ is representable
  if and only if it has a summation object $\Dbimon$ 
 such that for all objects $X, Y$,
  $\catLL(X, Y)$ has a PCM structure in which:
  \begin{enumerate}
    \item A family $(f_{i_1, \ldots, i_n})_{i_n, \ldots, i_n \in \N}$ of morphisms 
    $X \arrow Y$ is summable if and only if there exists 
    $h : X \sm \Dbimon^{\sm n} \arrow Y$ such that 
    $h \compl (X \tensor \Dinj_{i_n} \tensor \cdots \tensor \Dinj_{i_1}) 
    \compl (\tensorUnitR_X^n)^{-1}  = h_{i_1, \ldots, i_n}$.
     \item Then, 
    $\sum_{i_1, \ldots, i_n \in \N} h_{i_1, \ldots, i_n} 
    = h \compl (X \tensor \Ddiag \tensor \cdots \tensor \Ddiag) 
    \compl (\tensorUnitR_X^n)^{-1}$.
    \item The sum over the empty family is equal to the zero morphism $0$ 
  \end{enumerate}
  \end{proposition}

  It immediately follows from \ref{ax:unary} and from the characterization above 
  that there exists for all $i \in \N$ 
  a projection $\Dproj_i : \Dbimon \arrow 1$ such that $\Dproj_i \compl \Dinj_j 
  = \kronecker i j$ and $\Dproj_i \compl \Ddiag = \id_{1}$.
  As such, $\Dbimon$ is quite similar to a biproduct, except that 
  the pairing and the copairing do not always exist.

  \begin{theorem} \label{thm:representable-compatible}
  For all representable PCM category, the sum is compatible with the 
  monoidal product, the closure, and the cartesian product (whenever it exists).
  \end{theorem} 
  For proving \cref{thm:representable-compatible}, we check that 
  the summability structure $\Dbimon \linarrow \_$
  satisfies the conditions of
  \cref{prop:ss-ll} (see \cref{app:representable-summability}).
  
  We now adapt to our PCM setting one 
  of the main result of~\cite{EhrhardWalch25}: 
  there is a bijection between Taylor expansions 
  $\Sdl : \bang \S \naturalTrans \S \bang$
  and \emph{analytic coalgebras} on $\Dbimon$.
  
  The object $\Dbimon$ can be equipped with a
  comonoid structure, given by $\Dproj_0 : \Dbimon \arrow 1$ and 
  $\Dbimoncm : \Dbimon \arrow \Dbimon \sm \Dbimon$ 
  characterized by $\Dbimoncm \compl \Dinj_n = 
  \left(\sum_{i + j = n} (\Dinj_i \sm \Dinj_j)\right) \tensorUnitR_{1}^{-1}$.
  The morphism $\Dbimoncm$ exists because $(\Dinj_i \sm \Dinj_j)_{i, j \in \N}$ 
  is summable by \cref{prop:representable-uncurry}, so 
  $\left(\sum_{i+j =n} \Dinj_i \sm \Dinj_j\right)_{n \in \N}$ 
  is summable by \ref{ax:wpa}.
  The object $\Dbimon$ can also be equipped with a
  monoid structure, given by $\Ddiag : 1 \arrow \Dbimon$ and 
  $\Dbimonm : \Dbimon \sm \Dbimon \arrow \Dbimon$ characterized by 
  $\Dbimonm \compl (\Dinj_i \sm \Dinj_j) = \kronecker i j \compl \Dinj_i$.
  The morphism $\Dbimonm$ exists because the family 
  $(\Dinj_i)_{i \in \N}$ is summable by \cref{prop:representable-uncurry},
  so $(\kronecker i j \Dinj_i)_{i,j \in \N}$ is summable 
  by \cref{prop:zero-neutral}.

  \begin{theorem}[Theorem 212 of \cite{EhrhardWalch25}]
  The tuple $(\Dbimon, \Ddiag, \Dbimonm, \Dproj_0, \Dbimoncm)$
  is a bicommutative bimonoid.
  \end{theorem}

  The proof of \cite{EhrhardWalch25} also works for our PCMs.
  The bimonoid structure on $\Dbimon$ is related to the bimonad structure 
  of $\S$ through the 
  \emph{mate isomorphism}~\cite{KellyStreet74} (see~\cref{fig:mate}).  
  The bimonad $\S$ can be seen as a kind of writter/reader bimonad 
  associated to the bimonoid $\Dbimon$.

  \begin{definition}[Definition 218 of \cite{EhrhardWalch25}]
  An analytic coalgebra is a $\bang$-coalgebra $\Dbimonca : \Dbimon \arrow 
  \bang \Dbimon$ such that the bimonoid structure 
  of $\Dbimon$ and $\Dinj_0 : 1 \arrow \Dbimon$ 
  are coalgebra morphisms.
  \end{definition}

  \begin{theorem}[Corollary 222 of \cite{EhrhardWalch25}]
    There is a bijection between analytic Taylor expansions 
    $\Sdl : \bang \S \naturalTrans \S \bang$ for $\S = \Dbimon \linarrow \_$ 
     and analytic coalgebras.
  \end{theorem}

  The proof of \cite{EhrhardWalch25} also carries directly to our setting (it does 
  not use any assumption on summability, only a purely categorical 
  construction involving the mate isomorphism).

  \section{Taylor expansion in web models}

  \label{sec:web-model-representable}

  We prove in this section that our web models are representable, and feature 
  a Taylor expansion.
  For all set $A$ and all PCR $\rig$, define $\Ddiag_A \in \rig^A$ by $(\Ddiag_A)_a = 1$. 
  We then define  $\lone[A] \subseteq \rig^A$ and 
  $\linf[A] \subseteq \rig^A$ by
  $\lone[A] = \orth{\set{\Ddiag_A}}$ and $\linf[A] = \biorth{\set{\Ddiag_A}}$.
  In particular, 
  $\family{x_a} \in \lone[A]$ if and only if $\family{x_a}$ is summable and 
  $\sum_{a \in A} x_a \in \ball$.
  Then, we set $\Dbimon = (\N, \linf[\N])$.
  
  \begin{example}
  \begin{itemize}
  \item In probabilistic coherence spaces \cite{DanosEhrhard11}, $\lone[A]$ 
  is the set of all sub probability distributions 
  on $A$, and $\linf[A]$ is the set $[0,1]^A$.
  \item In finiteness spaces \cite{Ehrhard05}, $\lone[A]$ is the set of all $A$-indexed families 
  with finite support, and $\linf[A]$ is the set of all $A$-indexed families.
  \item In Köthe spaces \cite{Ehrhard02}, $\lone[A]$ is the set of all absolutely convergent 
  families, and $\linf[A]$ is the set of all bounded families.
  \end{itemize}
  More often than not, $\linf[A] = \points{\left(\withFam<A>[a] 1 \right)}$
  so $\Dbimon = \withFam<\N> 1$, as assumed in \cite{EhrhardWalch25}.
  However, this is not true in Köthe spaces in which 
  $\points{\left(\withFam<A>[a] 1 \right)}$ is the set of all $A$-indexed families.  
  \end{example}

  The object $\Dbimon$ is a summation object, with diagonal $\Ddiag = \Ddiag_{\N}$ and
  basis $\Dinj_i \in \rig^{\Web \Dbimon}$ defined by $(\Dinj_a)_{a'} = \kronecker a {a'}$ 
  (they belong to $\points \Dbimon$ by downward closure, since 
  $\Dinj_a \leq \Ddiag$).
  We now prove that the summability in $\WEB$ described in \cref{sec:summability-web}
  is representable. 

  \begin{lemma} \label{prop:Dbimon-iterated}
    For all $k \in \N$, $\points{\left(\Dbimon^{\sm k}\right)} = \linf[\N^k]$.
  \end{lemma}

  \begin{proof} 
    $\points{\left(\Dbimon^{\sm k}\right)} = 
  \biorth{\left(\biorth{\set{\Ddiag}} \sm \cdots \sm \biorth{\set{\Ddiag}}\right)} 
  \eqnum{1}= \biorth{\left(\set{\Ddiag} \sm \cdots \sm \set{\Ddiag}\right)} 
  = \biorth{\set{\Ddiag_{\N^k}}} = \linf[\N^k]$
  where \texteqnum{1} is obtained by iterating
  \cref{cor:predual-characterization}.
  \end{proof}

  \begin{theorem}
  The sum in $\WEBPOS$ and $\WEB$ is representable.
  \end{theorem}

  \begin{proof}
    It suffices to do the proof for $\WEB$ ($\WEBPOS$ is a special instance 
    of $\WEB$ in which $\abs{.}$ is the identity function).
  We prove that $\Dbimon$ satisfies the conditions of \cref{prop:representable-uncurry}.
  Let $(f(i_1, \ldots, i_k))_{i_1, \ldots, i_k \in \N}$ be a family of 
  $\semimod{X \linarrow Y}$.
  Define $h \in \vrig{X \sm \Dbimon^{\sm k} \linarrow Y}$ by 
  $h_{a, i_k, \ldots, i_1, b} = f(i_1, \ldots, i_k)_{a, b}$
  so that $f(i_1, \ldots, i_k) = h \cdot (X \sm \Dinj_{i_k} \sm \cdots \sm \Dinj_{i_1}) \cdot 
  (\tensorUnitR^{k})^{-1}$ 
  and $\sum_{i_1, \ldots, i_k \in \N} f(i_1, \ldots, i_k) 
  \sumiff h \cdot (X \sm \Ddiag \sm \cdots \sm \Ddiag) \cdot (\tensorUnitR^{k})^{-1} 
  \sumiff h \cdot (X \sm \Ddiag_{\N^k}) \cdot (\tensorUnitR^{k})^{-1}$.
  It suffices to prove that $(f(i_1, \ldots, i_k))_{i_1, \ldots, i_k \in \N}$
  is summable in $\semimod{X \linarrow Y}$ if and only if 
  $h \in \semimod {X \sm \Dbimon^{\sm k} \linarrow Y}$
  to conclude. We proceed by the following chain of equivalences.
  \begin{align*}
    &(f_{i_1, \ldots, i_k})_{i_1, \ldots, i_k \in \N} 
      \text{ is summable in } \semimod{X \linarrow Y} \\
    &\iff (\pos{f}_{i_1, \ldots, i_k})_{i_1, \ldots, i_k \in \N} 
    \text{ is summable in } \points{(X \linarrow Y)} \\ 
    &\eqnum{1}\iff \forall x \in \points X, \forall y' \in \points {\orth Y}, 
    \scalar{\sum_{i_1, \ldots, i_k \in \N} \pos{f}_{i_1, \ldots, i_k}}{x \sm y'}
    \text{is defined and in } \ball \\ 
    &\eqnum{2}\iff \forall x \in \points X, \forall y' \in \points {\orth Y}, 
    \scalar{\pos{h}}{x \sm \Ddiag_{\N^k} \sm y'} \text{ is defined and in } \ball \\ 
    &\eqnum{3}\iff \pos{h} \in \points{(X \sm \Dbimon^{\sm k} \linarrow Y)}
    \iff h \in \semimod{X \sm \Dbimon^{\sm k} \linarrow Y}
  \end{align*}
  where \texteqnum{1} is the predual characterization of summability given 
  in \cref{prop:summability-predual}, \texteqnum{2} holds by the covering 
  principle of \cref{prop:covering-principle-sm} ($\set{\Ddiag_{\N}}$ is 
  a covering), and \texteqnum{3} holds by 
  the predual characterization of \cref{thm:predual-characterization}, using 
  that $\points{\left(\Dbimon^{\sm k}\right)} = \biorth{\set{\Ddiag_{\N^k}}}$ 
  by \cref{prop:Dbimon-iterated}.
  \end{proof}

  Finally, the analytic coalgebra on 
  $\Dbimon$ is given by 
  \(\Dbimonca \in \vrig{\Dbimon \linarrow \bang \Dbimon}\) with  
  $\Dbimonca_{i, [i_1, \ldots, i_n]} = \kronecker{i}{i_1 + \cdots + i_n}$, so that
  $(\Dbimonca \cdot x)_{[i_1, \ldots, i_n]} = x_{i_1 + \cdots + i_n}$.
  %
  Independently of any specificities of the model, we always 
  have that 
  $\pos{\Dbimonca} \cdot \pos{\Ddiag} 
  = \pos{\prom{\Ddiag}} \in \points{(\bang \Dbimon)}$ so by 
  \cref{thm:predual-characterization}, $\Dbimonca \in \semimod{\Dbimon \linarrow 
  \bang \Dbimon}$.
  Thus, this analytic coalgebra generalizes to our generic 
  web models the analytic coalgebra of the 
  weighted relational model~\cite{EhrhardWalch25}. 
  This analytic coalgebra induces a Taylor expansion $\Sdl : \bang \S \naturalTrans 
  \S \bang$ that in turns induces a functor $\T$ on $\klexp$.
  For all $s \in \semimod{\bang X \linarrow Y}$, 
  $\T s \in \vrig{\bang \S X \linarrow \S Y}$ is
  defined as in the examples of \cite{EhrhardWalch25} by
  \begin{equation} \label{eq:taylor-functor}
    (\T s)_{[(i_1, a_1), \ldots, (i_k, a_k)], (j, b)} 
  = \kronecker{i_1 + \cdots + i_k}{j} \frac{\factorial{[a_1, \ldots, a_k]}}{
    \factorial{[(i_1, a_1), \ldots, (i_k, a_k)]}} s_{[a_1, \ldots, a_k], b}.
  \end{equation}
  We know \emph{by construction} that 
  $\T s$ belongs to $\semimod{\bang \S X \linarrow \S Y}$, for \emph{all} our web models.
  It is proved in \cite{EhrhardWalch25} that this functor 
  corresponds formally to the intuition given in 
  \cref{eq:coherent-Tinf-explicit}, using the 
  differential linear logical structure 
  of the weighted relational model~\cite{BluteCockettSeely06}.

\section{Conclusion and future work}

In~\cref{part:web-models}, we have built a generic construction of 
web models, based on strong and absolute PCM. It unifies various models of the literature. 
In \cref{part:taylor}, we have described a theory of Taylor expansion 
in categories enriched over PCMs. Web models given by our generic construction feature 
such Taylor expansion as proved in \cref{sec:web-model-representable}. 
Thus, Taylor expansion is compatible 
with the computational properties described by partial sums, such as determinism,
finite or countable non-determinism, randomness,
and absolute convergence.

\textbf{Future work.}
We conjecture that the exponential of our web models is always Lafont (see 
\cref{rem:lafont}). 
Our theory of coherent Taylor expansion 
captures a wide class of web models, but there are no known web-free examples. 
We conjecture that the analytic version~\cite{KerjeanTasson18} 
of the category of convenient 
vector spaces \cite{BluteEhrhardTasson10} is not representable but 
features a coherent Taylor expansion. 
%
%
Finally, although we described our web models with a strong emphasis 
on the \emph{algebraic} point of view of partial sums,
it should be possible to adopt a topological viewpoint on absolute PCMs, as it is done  in
Köthe spaces~\cite{Ehrhard02}
and finiteness spaces~\cite{Ehrhard05}.

\bibliography{biblio}

\appendix 

\section{Proofs on web based models}

\label{app:web-models}

\label{app:web-models-positive}

\scalarrearanging*

\begin{proof}
Observe that 
\[ \scalar{s \cdot x}{y}
= \sum_{b \in \web Y} \left(\sum_{a \in \web X} s_{a, b} x_a \right) y_b 
\eqnum{1}\sumsub \sum_{b \in \web Y} \sum_{a \in \web X} (s_{a, b} x_a y_b) 
\eqnum{2}\sumiff \ \smashoperator{\sum_{a \in \web X, b \in \web Y}} \ s_{a, b} x_a y_b 
=  \scalar{s}{x \sm y}\]
where \texteqnum{1} is by distributivity of multiplication over sums, and 
\texteqnum{2} is by \ref{ax:pa}. Furthemore, if 
$s \cdot x$ is defined, then \texteqnum{1} turns into an equivalence 
$\sumiff$ (but not in general, consider $y_b = 0$).
The second item is proved in the same way.
\end{proof}

\coveringtwo*

\begin{proof}
  Let $b_0 \in \web Y$. 
  We prove that $(s \cdot x)_{b_0}$ is defined.
  By covering, there exists $y \in F$ such that $y_{b_0}$ 
  is invertible. By assumption, $(s_{a, b} x_a y_b)_{a \in \web{X}, b \in \web Y}$
  is summable, so by \cref{prop:subfamily-summable}
  $(s_{a, b_0} x_a y_{b_0})_{a \in \web X}$ is summable.
  But then
  $y_{b_0}^{-1} \left(\sum_{a \in \web X} s_{a, b_0} x_a y_{b_0} \right)
  \sumsub \sum_{a \in \web X} s_{a, b_0} x_a y_{b_0} y_{b_0}^{-1} 
  = \sum_{a \in \web X} s_{a, b_0} x_a$ so 
  $(s_{a, b_0} x_a)_{a \in \web X}$ is summable.
\end{proof}

\linarrowcharacterization*

\begin{proof}
We first prove $(1) \iff (2)$. By \cref{rem:linarrow-predual},
\begin{align*} 
  s \in \points{(X \linarrow Y)} 
  &\iff \forall x \in \points X, \forall y' \in \points Y, 
  \scalar{s}{x \sm y'} \text{ is defined and in $B$.} \\
  \orth{s} \in \points{(\orth Y \linarrow \orth X)}
  &\iff \forall x \in \points X, \forall y' \in \points Y, 
  \scalar{\orth s}{y' \sm x} \text{ is defined and in $B$.} 
\end{align*}
We can easily check that $\scalar{s}{x \sm y'} \sumiff 
\scalar{\orth{s}}{y' \sm x}$, which concludes the proof of equivalence.
Next, we prove $(1) \imply (3)$. Let $s \in \points{(X \linarrow Y)}$ and 
$x \in \points X$. First, observe that 
$s \cdot x$ is well-defined by covering principle 
(\cref{prop:covering-principle-sm}), since 
$\points \orth{Y}$ is a covering.
Then, for all $y' \in \orth{Y}$, 
$\scalar{s \cdot x}{y'} \sumiff \scalar{s}{x \sm y'}$ by 
\cref{prop:scalar-rearanging}. Since
$s \orthrel {x \sm y'}$ by assumption, we get that $s \cdot x \orthrel y'$.
It concludes the proof of $(3)$. 
The proof of the reverse implication $(3) \imply (1)$ relies on 
\cref{prop:scalar-rearanging} in a similar same way. 
Finally, $(2) \iff (4)$ holds by applying the equivalence
$(1) \iff (3)$ on $\orth{s}$.
\end{proof}

\compositionwebpos*

\begin{proof}
Let $s \in \points{(X \linarrow Y)}$ and $t \in \points{(Y \linarrow Z)}$.
First, let us prove that $t \cdot s$ is well-defined. By covering, for all 
$a_0 \in \web X$ there exists $x \in \points X$ such that $x_{a_0}$ is invertible.
Then, $(t \cdot (s \cdot x))$ is defined by iterating 
\cref{prop:linarrow-characterization} twice. But
\[ (t \cdot (s \cdot x))_{c} 
= \sum_{b \in \web Y} t_{b, c} \left(\sum_{a \in \web X} s_{a, b} x_a \right)
\sumsub \sum_{b \in \web{Y}} \sum_{a \in \web X} t_{b, c} s_{a, b} x_a 
\overset{\text{\ref{ax:pa}}}{\sumiff} \ 
\smashoperator{\sum_{a \in \web{X}, b \in \web{Y}}} \ s_{a, b} t_{b, c} x_a. \]
It follows from \ref{ax:wpa} that $(s_{a_0, b} t_{b, x} x_{a_0})_{b \in \web Y}$ is summable,
so by invertibility of $x_{a_0}$, $(s_{a_0, b} t_{b, c})_{b \in \web Y}$ 
is summable.
That is, $t \cdot s$ is well-defined. We now prove that 
it belongs to $\points{(X \linarrow Z)}$ using \cref{prop:linarrow-characterization}.
Let $x \in \points X$. Then,
\[ ((t \cdot s) \cdot x)_c 
= \sum_{a \in \web X} \left(\sum_{b \in \web Y} s_{a, b} t_{b, c} \right) x_a 
\eqnum{1}{\sumiff} \sum_{a \in \web X} \left(\sum_{b \in \web Y} s_{a, b} t_{b, c} x_a \right). \] 
\[ (t \cdot (s \cdot x))_c 
= \sum_{b \in \web Y} t_{b, c} \left( \sum_{a \in \web X} s_{a, b} x_a \right) 
\eqnum{2}{\sumiff}  \sum_{b \in \web{Y}} \left( \sum_{a \in \web X} s_{a, b} t_{b, c} x_a \right). \]
Be aware: \texteqnum{1} and \texteqnum{2}
hold because we already know that $t \cdot s$ and $s \cdot x$ are well-defined.
Otherwise, it would not be an equivalence.
Then, by \ref{ax:pa}, $(t \cdot s) \cdot x \sumiff t \cdot (s \cdot x)$.
We conclude that $t \cdot s$ is well-defined 
and belongs to $\points{(X \linarrow Z)}$ by 
iterating \cref{prop:linarrow-characterization} twice.
\end{proof}

\predualcharacterization*

First, observe that $(1) \imply (2)$ is trivial.

\begin{proof}[Proof that $(2) \imply (3)$]
  Let $s \in \orth{(P \sm Q)}$. Then for all 
  $x \in P$, $s \cdot x$ is well defined by 
  covering principle (\cref{prop:covering-principle-sm}).
  Furthermore, for all $y' \in Q$, 
  $\scalar{s}{x \sm y'} \sumiff \scalar{s \cdot x}{'y'}$ by 
  \cref{prop:scalar-rearanging}.
  But $s \orthrel x \sm y'$ so we deduce that $s \cdot x \orthrel y'$.
  Thus, $s \cdot x \in \orth Q = \points Y$.
\end{proof}

\begin{proof}[Proof that $(3) \imply (2)$]
  Assume that $s$ satisfies $(3)$. Then we prove 
  that $s \in \orth{(P \sm \points{\orth Y})} \subseteq \orth{(P \sm Q)}$. 
  Let $x \in P$ and $y' \in \orth{(\points Y)}$.
  Then by assumption $s \cdot x$ is defined 
  so by \cref{prop:scalar-rearanging} $\scalar{s}{x \sm y'} \sumiff 
  \scalar{s \cdot x}{y'}$ and we have that
  $s \orthrel x \sm y'$. Thus, $s \in \orth{(P \sm \orth{\points Y})}$.
\end{proof}

  The equivalence $(2) \Leftrightarrow (4)$ holds by 
  a similar argument, since 
  $\scalar{s}{x \sm y'} \sumiff \scalar{\orth s}{y' \sm x}$.

\begin{proof}[Proof that $(3) \wedge (4) \imply (1)$]
Assume that $s$ satisfies $(3)$ and $(4)$. 
First, we know from the proof 
of $(3) \imply (2)$ that $s \in \orth{(P \sm \orth{\points Y})}$.
Second, by the implication 
$(2) \imply (4)$ applied on $P$ and $\orth{\points Y}$, 
we know that for all $y' \in \points \orth{Y}$, 
$\orth s \cdot y'$ is defined and belongs to 
$\orth{\points X}$. So by the characterization of \cref{prop:linarrow-characterization},
we conclude that $s \in \points{(X \linarrow Y)}$.
\end{proof}

\section{Results on summability structures}

\subsection{Comparison with the product and the coproduct}

\label{app:ss-coproduct-product}

We prove here that the
summability structure sits in between the coproducts and 
the product. As stated in \cref{prop:bimonad-structure}, 
for all $i \in \N$ and all object $X$
there exists a morphism \(\Sinj_i \in\catLL(X,\S X)\) such
that $\Sproj_j \compl \Sinj_i = \kronecker i j$. 
Furthermore, $\Ssum \compl \Sinj_i = \id$. This is an immediate 
consequence of \ref{ax:unary} and \ref{def:ss-sum}.
We now prove that there is a \emph{factorization} 
\begin{equation} \label{eq:factorization-ss}
\begin{tikzcd}[column sep = large]
	{\bigoplus_{i=0}^n X } & {\S X} & {\prod_{i = 0}^n X.}
	\arrow["{\coprodPairing<\N>{\Sinj_i}}", from=1-1, to=1-2]
	\arrow["{\prodPairing<\N>{\Sproj_i}}", from=1-2, to=1-3]
\end{tikzcd}
\end{equation}
Observe that  $\prodPairing<\N>{\Sproj_i}$ is a mono, this is an immediate
consequence of the joint monicity of the $\Sproj_i$ given
by \ref{def:ss-proj} and the uniqueness of the pairing. We now prove that 
$\coprodPairing<\N>{\Sinj_i} : \bigoplus_{i=0}^n X \arrow \S X$
is an epi. For a similar reason, it suffices to prove that 
the $\Sinj_i$ are jointly epic. This is done in \cref{prop:summability-of-pieces}
and \cref{prop:Sinj-jointly-epic} below.
  
\begin{proposition} \label{prop:summability-of-pieces}
The family $\sequence{\Sinj_i \compl \Sproj_i}$ of $\catLL(\S X, \S X)$ is summable, 
with sum equal to $\id_{\S X}$.
\end{proposition}

\begin{proof}
Observe that
$\Sproj_j \compl \Sinj_i \compl \Sproj_i = \kronecker i j \ \Sproj_i$. 
It follows from the summability of $\sequence{\Sproj_i}$ and 
\cref{prop:zero-neutral} that
$(\Sproj_j \compl \Sinj_i \compl \Sproj_i)_{i, j \in \N}$ is 
summable. But then, $\Sproj_j \compl \Sproj_i \compl 
\Spairing[(j,i)]<\N^2>{\Sproj_j \compl \Sinj_i \compl \Sproj_i}
=  \Sproj_j \compl \Sinj_i \compl \Sproj_i$ so by joint monicity
of the $\Sproj_j$, 
$\Sproj_i \Spairing[(j,i)]<\N^2>{\Sproj_j \compl \Sinj_i \compl \Sproj_i}
= \Sinj_i \compl \Sproj_i$ and
$\sequence{\Sinj_i \compl \Sproj_i}$ is summable by \ref{def:ss-sum}.
Then by \cref{prop:ss-is-pcm-cat},
\[ \Sproj_k \compl \left(\sum_{i \in \N} \Sinj_i \compl \Sproj_i \right) 
\sumsub \sum_{i \in \N} (\Sproj_k \compl \Sinj_i \compl \Sproj_i)
\sumsub \sum_{i \in \N} (\kronecker i k \ \Sproj_i)
= \Sproj_k \text{ by \ref{ax:unary}} \]
so by joint monicity of the $\Sproj_k$, $\sum_{i \in \N} \Sinj_i \compl \Sproj_i 
= \id_{\S X}$. 
\end{proof}

\begin{corollary} \label{prop:Sinj-jointly-epic}
The $\Sinj_i$ are jointly epic: for all 
$f, g : \S X \arrow Y$, if $f \compl \Sinj_i = g \compl \Sinj_i$ for 
all $i \in \N$, then $f = g$.
\end{corollary}

\begin{proof}
By a straightforward computation,
\[ f = f \compl \left( \sum_{i \in \N} \Sinj_i \compl \Sproj_i \right)
\sumsub \sum_{i \in \N} f \compl \Sinj_i \compl \Sproj_i 
= \sum_{i \in \N} g \compl \Sinj_i \compl \Sproj_i 
\sumsubinv g \compl \left( \sum_{i \in \N} \Sinj_i \compl \Sproj_i \right)
= g \qedhere \]
\end{proof}

When the category has arbitrary sums (either finite or countable), 
then the hierarchy of \cref{eq:factorization-ss} collapses into 
an isomorphism bewteen the coproduct, the product, and 
the summability structure, and the category has countable 
biproducts. 
This is an infinitary counterpart of a well known
result on finite biproducts, see for example Chapter 8 of \cite{Maclane71}.
Observe the similarity between the proof of this result and 
the proof of \cref{prop:Sinj-jointly-epic}.

\begin{proposition}[Proposition 237 of \cite{EhrhardWalch25}] 
  If the PCM category $\catLL$ is total (the PCM structure on each homset is 
  a total function) then it has countable biproducts.
\end{proposition}

\begin{proof}
Having $I$-indexed biproducts means that the cartesian product 
$\prod_{i \in I} \_$ is also a coproduct.
We define the injections $\inj_i : X_i \arrow \prod_{i \in I} X_i$ 
by $\Sinj_i = \prodPairing[j]{\kronecker i j \id_X}$.
The universal property of the coproduct is proved using the fact that 
\[ \sum_{i \in} \Sinj_i \compl \proj_i = \id. \qedhere \]
\end{proof}

\subsection{Compatibility with the categorical structure}

\label{app:compatibility-pcm-ll}

\compatibilitypcmll*

\begin{proof}[Proof of Item 1]
  The forward implication is trivial, because \cref{eq:sm-strength}
  is a particular instance of \cref{def:sum-sm}. The proof 
  of the reverse direction is very similar 
  to the proof of \cref{prop:ss-is-pcm-cat}.
  Let $g : X_2 \arrow Y_2$ be a morphism, and
  $(f_i)_{a \in A}$ be a summable 
  family of morphisms in $X_1 \arrow X_2$.  Let 
  $\psi : A \injection \N$ be an injection.
  Let $(f_i')_{i \in \N} = \famact{\psi}(f_a)_{a \in A}$. This family is summable by 
  \cref{prop:reindexing}.
  Let 
  \[ h = \SstrR_{Y_1, Y_2} \compl \left(\Spairing{f_i'} \tensor g\right)
  : X_1 \tensor Y_2 \arrow \S(Y_1 \tensor Y_2). \]
  
  Then we can easily check (by definition of $\SstrR$) 
  that $\Sproj_i \compl h = f_i' \tensor g$ 
  and $\Ssum \compl h = \left(\sum_{i \in \N} f_i' \right) \tensor g$,
  so by \ref{def:ss-sum} $(f_i')_{i \in \N}$ is summable 
  with sum $\left(\sum_{a \in A} f_a \right) \tensor g$.
  We conclude that $(f_a)_{a \in A}$ is summable with the same sum, 
  by reindexing (\cref{prop:reindexing}), and using 
  that $0 \tensor g = 0$. 
  So $\sm$ is left distributive. We conclude that $\sm$ is strongly distributive
  by \cref{rem:sum-ll-weaker}.
\end{proof}

  \begin{proof}[Proof of Item 2]
  For the forward direction, recall that $X \linarrow \Sproj_i :
  (X \linarrow \S Y) \arrow (X \linarrow Y)$
  is defined as $\cur\left(\ev \compl \left( (X \linarrow Y) \sm \Sproj_i\right) \right)$.
  But $\ev \compl ((X \linarrow Y) \sm \Sproj_i)$ is summable, by summability
  of the $\Sproj_i$, right distributivity (of $\sm$) and right distributivity
  (of composition). So the $X \linarrow \Sproj_i$ are summable.
  
  Conversely, it suffices by \cref{rem:sum-ll-weaker}
  to prove that for all summable family $\family{f_a}$ of $\catLL(X \sm Y, Z)$, 
  $\family{\cur(f_a)}$ is summable.
  Let $\psi : A \injection \N$ be an injection.
  Let $(f_i')_{i \in \N} = \famact{\psi}(f_a)_{a \in A}$. This family is summable by 
  \cref{prop:reindexing}. Let 
  \[ h = \Spairing{Y \linarrow \Sproj_i} \compl \cur \left(\Spairing{f_i'}\right) : 
  X \arrow \S(Y \linarrow Z). \]
  Then we can check that $\Sproj_i \compl h = \cur(f_i')$, 
  so $(\cur(f_i'))_{i \in \N}$ is summable. But 
  $(\cur(f_i'))_{i \in \N} = \famact{\psi}(\cur(f_a))_{a \in A}$
  (because $\cur(0) = 0$) so we conclude that 
  $\family{\cur(f_a)}$ is summable.
  \end{proof}

\begin{proof}[Proof of Item 3]
The forward direction is immediate, upon observing that 
$\withFam \Sproj_j = \prodPairing{\Sproj_j \comp \proj_i}$
and that for all $i \in I$, 
$(\Sproj_j \compl \proj_i)_{j \in \N}$ is summable 
(by right distributivity of composition).
The reverse direction is very similar to the proof of Item 2.
\end{proof}

\section{Results on representable summability}
\label{app:representable-summability}
The goal of this section is to prove \cref{thm:representable-compatible}.
This is a consequence of 
\cref{lemma:representable-sm,lemma:representable-closure,lemma:representable-product}
proved below.

\begin{lemma} \label{lemma:representable-sm}
The sum of a representable PCM category is always compatible with 
the symmetric monoidal product.
\end{lemma}

\begin{proof}
  We rely on the characterization of \cref{prop:ss-ll}.
  We already know by assumption that $0$ is absorbing for the 
  monoidal product. Let us prove that $\sequence{\Sproj_i \tensor X_2}$ is 
  a summable family of $\catLL((\Dbimon \linarrow X_1) \sm X_2, X_1 \sm X_2)$,
  with sum $\Ssum \tensor X_1$. Let
    \begin{equation} \label{eq:SstrR-representable-uncurry}
        f = \begin{tikzcd}[column sep = large]
          {(\Dbimon \linarrow X_1) \tensor X_2 \tensor \Dbimon} & {(\Dbimon \linarrow X_1) \tensor \Dbimon \tensor X_2} & {X_1 \tensor X_2}
          \arrow["{(\Dbimon \linarrow X_1) \tensor \tensorSym}", from=1-1, to=1-2]
          \arrow["{\ev \tensor X_2}", from=1-2, to=1-3]
        \end{tikzcd}.
      \end{equation}
        Then we can check that 
        \begin{align*}
        f \compl ((\Dbimon \linarrow X_1) \tensor X_2 \tensor \Dinj_i) \compl 
        \tensorUnitR^{-1}_{\Dbimon \linarrow X_1 \tensor X_2}
        &\eqnum{1}= (\ev \tensor X_2) \compl ((\Dbimon \linarrow X_1) \tensor \Dinj_i \tensor X_2)
        \compl (\tensorUnitR^{-1}_{\Dbimon \linarrow X_1} \tensor X_2) \\
        &\eqnum{2}= \Sproj_i \tensor X_2
        \end{align*}
        where \texteqnum{1} is obtained by standard computations on 
        symmetric monoidal categories, and \texteqnum{2}
        is an immediate consequence of \cref{eq:projection-representable-def}.
        So by \cref{prop:representable-uncurry},
        $\sequence{\Sproj_i \tensor X_2}$ is summable with sum
        \[ f \compl ((\Dbimon \linarrow X_1) \tensor X_2 \tensor \Ddiag) \compl 
        \tensorUnitR^{-1}_{\Dbimon \linarrow X_1 \tensor X_2} 
        = \Ssum \tensor X_2 \] 
        by a similar computation.
        The summability of $\sequence{X_1 \tensor \Sproj_i}$ follows by symmetry 
        of the monoidal product, as explained in \cref{rem:sum-ll-weaker}.
\end{proof}

\begin{lemma} \label{lemma:representable-closure}
The sum of a representable PCM category is always compatible with 
the symmetric monoidal closed structure.
\end{lemma}

\begin{proof}
  By \cref{prop:ss-ll} is suffices to prove that 
  $\sequence{X \linarrow \Sproj_i}$ is a 
  summable family of 
  $\catLL(X \linarrow (\Dbimon \linarrow Y), X \linarrow Y)$. Consider the morphism
  \[ \cur(h) : (X \linarrow (\Dbimon \linarrow Y)) \sm \Dbimon \arrow (X \linarrow Y) \]
  where $h$ is defined as 
  \[ 
  \begin{tikzcd}[column sep = small]
	  {(X \linarrow (\Dbimon \linarrow Y)) \sm \Dbimon \sm X} & {(X \linarrow (\Dbimon \linarrow Y)) \sm X \sm \Dbimon} & {(\Dbimon \linarrow Y) \sm \Dbimon} & Y
	  \arrow["{\id \sm \smsym}", from=1-1, to=1-2]
	  \arrow["{\ev \sm \Dbimon}", from=1-2, to=1-3]
	  \arrow["\ev", from=1-3, to=1-4]
  \end{tikzcd}\] 
Then 
\begin{align*}
&\cur(h) \compl \Big( \big(X \linarrow (\Dbimon \linarrow Y) \big) \sm \Dinj_i\Big) \compl 
  \tensorUnitR^{-1}_{X \linarrow (\Dbimon \linarrow Y)} \\
&= \cur\bigg(h \compl \Big(\big(X \linarrow (\Dbimon \linarrow Y)\big) \sm \Dinj_i \sm X\Big)
\compl \Big(\tensorUnitR^{-1}_{X \linarrow (\Dbimon \linarrow Y)} \sm X\Big) \bigg) 
\tag*{By naturality of $\cur$} \\
&= \cur(\Sproj_i \compl \ev_{X, \Dbimon \linarrow Y}) \tag*{By the diagram chase below} \\ 
&= X \linarrow \Sproj_i \tag*{By definition of $\linarrow$}
\end{align*}
\[ 
\begin{tikzcd}[column sep = small]
	{(X \linarrow (\Dbimon \linarrow Y)) \sm \Dbimon \sm X} & {(X \linarrow (\Dbimon \linarrow Y)) \sm X \sm \Dbimon} & {(\Dbimon \linarrow Y) \sm \Dbimon} & Y \\
	{(X \linarrow (\Dbimon \linarrow Y)) \sm 1 \sm X} & {(X \linarrow (\Dbimon \linarrow Y)) \sm X \sm 1} & {(\Dbimon \linarrow Y) \sm 1} \\
	{(X \linarrow (\Dbimon \linarrow Y)) \sm X} & {(X \linarrow (\Dbimon \linarrow Y)) \sm X} & {\Dbimon \linarrow Y}
	\arrow["{\id \sm \smsym}", from=1-1, to=1-2]
	\arrow["h"{pos=0.6}, shift left=3, curve={height=-12pt}, from=1-1, to=1-4]
	\arrow["{\ev \sm \Dbimon}", from=1-2, to=1-3]
	\arrow["\ev", from=1-3, to=1-4]
	\arrow["{\id \sm \Dinj_i \sm X}", from=2-1, to=1-1]
	\arrow["{\id \sm \smsym}", from=2-1, to=2-2]
	\arrow["{\id \sm \id \sm \Dinj_i}", from=2-2, to=1-2]
	\arrow["{\ev \sm 1}", from=2-2, to=2-3]
	\arrow["{\id \sm \Dinj_i}", from=2-3, to=1-3]
	\arrow["{\cref{eq:projection-representable-def}}"{description}, draw=none, from=2-3, to=1-4]
	\arrow["{\tensorUnitR^{-1} \sm X}", from=3-1, to=2-1]
	\arrow["\id"', from=3-1, to=3-2]
	\arrow["{\tensorUnitR^{-1}}"', from=3-2, to=2-2]
	\arrow["\ev"', from=3-2, to=3-3]
	\arrow["{\Sproj_i}"', shift right=3, curve={height=18pt}, from=3-3, to=1-4]
	\arrow["{\tensorUnitR^{-1}}", from=3-3, to=2-3]
\end{tikzcd}\]
By \cref{prop:representable-uncurry}, we conclude that 
the $X \linarrow \Sproj_i$ are summable.
\end{proof}

\begin{lemma} \label{lemma:representable-product}
The sum of a representable PCM category is always compatible
with the categorical products, whenever they exist.
\end{lemma}

\begin{proof}
Let $I$ be a set such that $I$-indexed products exist.
By \cref{prop:ss-ll}, it suffices to prove that 
$\sequence<j>{\withFam \Sproj_j}$ is summable.
     Define
     \[ f = \prodPairing{\ev \compl (\proj_i \tensor \Dbimon)} : 
     \left(\withFam (\Dbimon \linarrow X_i) \right) \tensor \Dbimon \arrow
       \withFam X_i \]
     Then, 
     \begin{align*}
     f \compl (\withFam X_i \tensor \Dinj_j) \compl \tensorUnitR^{-1} 
     &= \prodPairing{\ev \compl (\proj_i \tensor \Dinj_j) 
      \compl \tensorUnitR^{-1}} \\
     &= \prodPairing{\ev \compl ((\Dbimon \linarrow X) \tensor \Dinj_j) 
      \compl \tensorUnitR^{-1} \compl 
      \proj_i} \tag*{by functoriality and naturality} \\
     &= \prodPairing{\Sproj_j \compl \proj_i} \tag*{by \cref{eq:projection-representable-def}} \\
     &= \withFam \Sproj_j
     \end{align*}
Thus by \cref{prop:representable-uncurry}, $\sequence<j>{\withFam \Sproj_j}$ 
is summable.
\end{proof}

\begin{figure}
  \begin{align*}
    \intertext{Monoidal structure: $\smone = (\set{*}, \biorth{\set{e_*}})$ and
    for all $s \in \points {(X_1 \linarrow Y_1)}$ and 
    $s' \in \points{(X_2 \linarrow Y_2)}$, 
    the tensor $s \sm s' \in \vrigpos{X_1 \sm X_2 \linarrow Y_1 \sm Y_2}$ 
    is defined by $(s \sm s')_{(a, a'), (b, b')} = s_{a, b} s'_{a', b'}$ 
    so that $(s \sm s') \cdot (x \sm x') = (s \cdot x) \sm (s' \cdot x')$}
    \smunitL \in \vrigpos{1 \sm X \linarrow X} 
    && \smunitL_{(*, a), a'} &= \kronecker a {a'} 
    & \smunitL \cdot (r \sm x) &= r x \\
    \smunitR \in \vrigpos{X \sm 1 \linarrow X} 
    && \smunitR_{(a, *), a'} &= \kronecker a {a'} 
    & \smunitR \cdot (x \sm r) &= r x \\
    \smassoc \in \vrigpos{(X \sm Y) \sm Z \linarrow X \sm (Y \sm Z)} 
    && \hspace{-0.5em} \smassoc_{((a, b), c), (a', (b', c'))} 
    &= \kronecker a {a'} \kronecker b {b'} \kronecker c {c'} \hspace{-0.5em}
    & \smassoc \cdot ((x \sm y) \sm z) &= x \sm (y \sm z) \\ 
    \smsym \in \vrigpos{X \sm Y \linarrow Y \sm X} 
    && \smsym_{(a, b), (b', a')} &= \kronecker a {a'} \kronecker b {b'}
    & \smsym \cdot (x \sm y) &= y \sm x
    \intertext{Closed structure: for all $s \in \points{(X \sm Y \linarrow Z)}$}
    \ev \in \vrigpos{((X \linarrow Y) \sm X) \linarrow Y} 
    && \ev_{((a, b), a'), b'} &= \kronecker a {a'} \kronecker b {b'} 
    & \ev \cdot (s \sm x) &= s \cdot x  \\ 
    \cur(s) \in \vrigpos{X \linarrow (Y \linarrow Z)}
    && \cur(s)_{a, (b, c)} &= s_{(a, b), c} 
    & (\cur(s) \cdot x) \cdot y &= s \cdot (x \sm y)
    \intertext{Exponential: for all $s \in \points{(X \linarrow Y)}$}
    \bang s \in \vrigpos{\bang X \linarrow \bang Y} 
    && (\bang s)_{m, [b_1, \ldots, b_n]} 
    &= \smashoperator{\sum_{\twolines{a_1, \ldots, a_n \in \web X}
      {\text{s.t. } m = [a_1, \ldots, a_n]}}}
      s_{a_1, b_1} \cdots s_{a_n, b_n} \hspace{-5em}
    & \bang s \cdot \prom{x} &= \prom{(s \cdot x)} \\ 
    \der \in \vrigpos{\bang X \linarrow X} 
    && \der_{m, a} &= \kronecker a {[a]} 
    & \der \cdot \prom{x} &= x \\ 
    \dig \in \vrigpos{\bang \bang X \linarrow \bang X} 
    && \hspace{-5em} \dig_{[m_1, \ldots, m_k], m} 
    &= \kronecker{m}{m_1 + \cdots + m_n}  \hspace{-5em}
    & \dig \cdot \prom{(\prom x)} &= \prom x
    \intertext{Seely isomorphisms: Define $k \cdot [a_1, \ldots, a_n]
    = [(k, a_1), \ldots, (k, a_n)]$. Observe that
    every $m \in \web{\bang{(X \with Y)}}$ can be written 
    uniquely as $m = 1 \cdot m_1 + 2 \cdot m_2$ with 
    $m_1 \in \web{\bang X}$ and $m_2 \in \web{\bang Y}$
    (even when taking into account that $\web{\bang X}$ is not the 
    set of all multisets in general). 
    Furthermore, recall that $\top = (\emptyset, \set{0})$
    so $\bang{\top} = \left(\set{[]}, \biorth{\set{\prom 0}}\right) 
    = \left(\set{[]}, \biorth{\set{e_{[]}}}\right)$.}
    \seelyzero \in \vrigpos{1 \linarrow \bang \top} 
    && \seelyzero_{*, []} &= 1 
    & \seelyzero \cdot r &= r \\
    (\seelyzero)^{-1} \in \vrigpos{\bang \top \linarrow 1} 
    && (\seelyzero)^{-1}_{[], *} &= 1 
    & (\seelyzero)^{-1} \cdot r &= r \\
    \seelytwo \in \vrigpos{(\bang X \sm \bang Y) \linarrow \bang(X \with Y)} 
    && \hspace{-5em} \seelytwo_{(m_1, m_2), m} 
    &= \kronecker {m} {1 \cdot m_1 + 2 \cdot m_2} \hspace{-5em}
    & \seelytwo \cdot (\prom{x} \sm \prom{y}) &= \prom{\prodPair{x}{y}} \\
    \hspace{-10em}(\seelytwo)^{-1} \in \vrigpos{\bang(X \with Y) \linarrow (\bang X \sm \bang Y)} 
    && \hspace{-5em} (\seelytwo)^{-1}_{m, (m_1, m_2)} 
    &= \kronecker {m} {1 \cdot m_1 + 2 \cdot m_2} \hspace{-5em}
    & (\seelytwo)^{-1} \cdot \prom{\prodPair x y} &= 
    \prom x \sm \prom y \\
  \end{align*}
  \caption{LL structure of $\WEBPOS$}
  \label{fig:web-ll}
\end{figure}

\begin{figure}
  \begin{center}
  \begin{tabular}{c c c c c}
    Bimonoid $\Dbimon$ & & Bimonad $\_ \tensor \Dbimon$ 
    & & Bimonad $\Dbimon \linarrow \_$ \\ \\

    Projection $\Dproj_i \in \categoryLL(\Dbimon, 1)$ & & 
    $\tensorUnitR \compl (X \tensor \Dproj_i) \in \categoryLL(X \tensor \Dbimon, X)$ 
    & & Injection $\Sinj_i$ \\
    Basis $\Dinj_i \in \categoryLL(1, \Dbimon)$ & & 
    $(X \tensor \Dinj_i) \compl \tensorUnitR^{-1} \in \categoryLL(X, X \tensor \Dbimon)$
    & & Projection $\Sproj_i$ \\ \\ 

    Monoid unit $\Dbimonu$ & & Monad unit $(X \tensor \Dbimonu) \compl \tensorUnitR^{-1}$
    & & Comonad unit $\Ssum$ \\
    Monoid multiplication $\Dbimonm$ & & Monad sum $(X \tensor \Dbimonm) \compl \tensorAssoc$
    & & Comonad sum $\Slift$ \\
    Comonoid unit $\Dbimoncu$ & \hspace{-3em} $\Longleftrightarrow$ \hspace{-3em} & 
    Comonad unit $\tensorUnitR \compl (X \tensor \Dbimoncu)$ 
    & \hspace{-3em} $\overset{\text{mates}}{\Longleftrightarrow}$ \hspace{-3em} 
    & Monad unit $\Sinj_0$ \\
    Comonoid multiplication $\Dbimoncm$ & & Comonad sum $\tensorAssoc \compl (X \tensor \Dbimoncm)$
    & & Monad sum $\SmonadSum$ \\
    Commutativity $\tensorSym$ & & Distributive law $X \tensor \tensorSym$
    & & Distributive law $\Sswap$ \\ \\

    Taylor coalgebra & & Distributive law & & Distributive law  \\
    $\Dbimonca \in \catLL(\Dbimon, \bang \Dbimon)$ & & 
    $(\bang\_) \tensor \Dbimon \naturalTrans \bang(\_ \tensor \Dbimon)$ & & 
    $\Sdl : \bang \S \naturalTrans \S \bang$
  \end{tabular}
\end{center}
\caption{Bimonoid and bimonad relations}
\label{fig:mate}
\end{figure}